\documentclass[11pt,a4paper]{article}
\usepackage[margin=1in]{geometry}

\usepackage{tikz}

\usepackage{lipsum}
\usepackage{amsfonts}
\usepackage{graphicx}
\usepackage{epstopdf}
\usepackage{algorithmic}
\ifpdf
  \DeclareGraphicsExtensions{.eps,.pdf,.png,.jpg}
\else
  \DeclareGraphicsExtensions{.eps}
\fi

\usepackage{enumitem}
\setlist[enumerate]{leftmargin=.5in}
\setlist[itemize]{leftmargin=.5in}

\usepackage{amsthm}
\newtheorem{theorem}{Theorem}
\newtheorem{proposition}{Proposition}
\newtheorem{corollary}{Corollary}

\newtheorem{definition}{Definition}

\usepackage{caption}
\usepackage{subcaption}
\usepackage{algorithm}

\usepackage{url}

\title{Random spanning tree Markov random field priors for Bayesian inverse problems in imaging}
\author{Jasper M. Everink}
\date{University of Eastern Finland, Kuopio, Finland}

\usepackage{amsopn}
\usepackage{amsmath}

\usepackage{float}

\usepackage{bm}
\renewcommand{\vec}[1]{\bm{#1}}

\begin{document}

\maketitle

\begin{abstract}
Markov random fields are common prior distributions used in Bayesian inverse imaging problems. In particular, difference priors assign probability distributions to differences between neighbouring pixels, such as Gaussian, Laplace, or Cauchy distributions. Depending on the chosen difference distribution, these priors have smoothing or edge-preserving properties. In this work, we propose a hyperprior on the connectivity graph of the pixel grid in the form of a random spanning tree, i.e., a random connected graph with the minimal number of edges, thereby coupling continuous and discrete random variables in the prior. By using random spanning trees, only a sparse random subset of edges is regularized, which helps preserve edges in the image with reduced contrast loss compared to standard difference-based Markov random fields. We discuss how fractal-like interfaces arise in high-resolution prior samples due to the random-tree connectivity. Finally, we propose a Gibbs sampler that alternates between the discrete tree updates and continuous pixel updates to efficiently explore the posterior distribution. We apply the method to various standard test image restoration problems, including denoising, deblurring, and inpainting, to study the impact of the proposed prior in comparison with existing Markov random fields.
\end{abstract}

\section{Introduction}\label{sec:introduction}

In the Bayesian formulation for finite-dimensional linear inverse problems \cite{kaipio2005statistical}, the observations $\vec{y} \in \mathbb{R}^m$ and unknown parameters $\vec{x} \in \mathbb{R}^n$ are assumed to be random vectors; they are related using a deterministic forward operator $A: \mathbb{R}^n \rightarrow \mathbb{R}^m$ and a random error $\vec{e} \in \mathbb{R}^m$ by the equation
\begin{equation}\label{eq:linear_inverse}
    \vec{y} = A\vec{x} + \vec{e}.
\end{equation}
To reconstruct $\vec{x}$, we can construct a posterior probability distribution on $\vec{x}$ given the random observations $\vec{y}$ by means of Bayes's rule,
\begin{equation}\label{eq:bayes}
    \pi(\vec{x}\,|\, \vec{y})\,\propto\, \pi(\vec{y}\,|\, \vec{x})\pi(\vec{x}),
\end{equation}
where the likelihood $\pi(\vec{y}\,|\, \vec{x})$ is derived from \eqref{eq:linear_inverse} and the prior $\pi(\vec{x})$ encodes any assumptions on the unknown $\vec{x}$ before any observations are considered. Similar to regularization in the variational formulation for inverse problems, the prior serves to penalize unwanted behaviour such as noise and promote desirable properties such as regularity and sparsity.

For imaging applications, many families of model-driven prior distributions have been studied. One such family is Markov random field (MRF) priors, which in this work are limited to difference priors defined by assigning independent, real-valued distributions to the differences between neighbouring pixels. This includes smoothness-promoting priors such as the Gaussian Markov random field (GMRF) priors \cite{rue2005gaussian}, also referred to as Gaussian free fields (GFFs) \cite{sheffield2007gaussian}. It also includes sparsity-promoting priors such as Laplace Markov random field (LMRF) priors \cite{bardsley2012laplace, siltanen2003statistical}, sometimes called total variation (TV) priors, Cauchy Markov random field (CMRF) priors \cite{markkanen2019cauchy, suuronen2022cauchy}, Student's t Markov random field priors \cite{senchukova2024bayesian}, and the horseshoe prior for edge-preserving inversion \cite{uribe2023horseshoe}. These sparsity-promoting priors often result in near-constant regions with discontinuous jumps between them. Strictly enforcing constant regions, whilst possible using spike-and-slab priors, can be very computationally expensive \cite{hans2007shotgun, schreck2015shrinkage}.

In general, MRFs are defined on graphs, and hence rely on the underlying graph structure. Graphs with simple structures, such as one-dimensional signals on line graphs, can help simplify the analysis and use of MRFs in Bayesian inversion. In particular, trees, i.e., graphs without loops, are well-studied and have many applications. For example, trees are used to infer structure in graphs \cite{duan2023bayesian}. They are also used to improve the efficiency of TV denoising \cite{kolmogorov2016total} and GMRF sampling \cite{brown2021sampling}. However, for Bayesian inverse problems in imaging, the graph is often assumed to be a grid of pixels and thereby fixed with limited flexibility.

The use of trees in prior distributions can lead to priors with rich structures. One natural way in which trees appear is in priors based on wavelet decompositions. In such priors, the hierarchical nature of wavelets naturally leads to a tree of coefficients. One such family is the Besov tree priors \cite{kekkonen2023random, kekkonen2026random}, which dynamically truncate the coefficients occurring in the wavelet representation of Besov priors \cite{horst2025uncertainty}. Deep branches in these Besov tree priors naturally correspond to highly detailed, fractal-like structures in the reconstruction.

Although grid-based MRFs are efficient, they lack adaptive or data-driven dependence between edges. This causes contrast loss if they are sufficiently strong in the case of light-tailed distributions, e.g., GMRFs and LMRFs. Alternatively, heavy-tailed distributions allow for larger jumps along edges, e.g., CMRFs, but result in distributions that can be multi-modal or not log-concave, which is possibly computationally inefficient. Meanwhile, existing tree-based approaches in Bayesian imaging are motivated from infinite dimensional priors defined on function-spaces and work on wavelet coefficients and not pixel adjacency.

\subsection{Contributions}
In this work, we propose the use of random spanning trees as a hyperprior on the graph structure of MRFs used in imaging. These random spanning trees enable efficient computation whilst providing a nontrivial probability distribution that can result in fractal-like interfaces in high-resolution images. Concretely, the major contributions in this work are as follows:

\begin{itemize}
    \item We introduce random spanning tree Markov random fields (RST-MRFs), which use random spanning trees as hyperpriors for standard Markov random field priors such as GMRFs, LMRFs, and CMRFs.
    \item We discuss how to efficiently sample from the associated posteriors using a Gibbs sampler that alternates between the continuous image variables and the discrete random spanning tree. We study the computational cost of the proposed methods and discuss strategies for balancing prior complexity and computational efficiency.
    \item We apply the RST-MRFs to standard test problems in Bayesian image reconstruction, including denoising, deblurring, and inpainting, to study the behaviour of the prior, its artifacts, and the associated computational cost.
\end{itemize}

\subsection{Overview}
This article is structured as follows. In Section \ref{sec:background}, we introduce the necessary notation and background on graphs, Markov random field priors, and random spanning trees needed to define and sample from the proposed RST-MRF priors. In Section \ref{sec:method}, we define the RST-MRF priors, present an algorithm for sampling from the posterior, examine fractal-like interfaces in high-resolution samples, and discuss various modeling and computational considerations. In Section \ref{sec:experiment}, we apply the RST-MRF priors and standard MRF priors to benchmark test problems to compare their behaviour. Finally, in Section \ref{sec:conclusions}, we provide some concluding remarks and discuss directions in which the RST-MRF priors can be further studied and developed.

\section{Background}\label{sec:background}

This section describes the necessary background on graphs, Markov random fields (including GMRFs, LMRFs and CMRFs), hierarchical modeling, and random spanning trees needed for the development of the method in Section \ref{sec:method}. Whilst the focus in this work is on imaging problems for Theorem \ref{thm:fractal_interface}, the method and all other theory presented can be applied to any inverse problem on graphs. Hence, we will use fully general notation.

\subsection*{Graph notation and conventions}
We denote a graph by $G = (V, E)$, where $V$ denotes the set of vertices and $E$ the set of edges between the vertices. We assume the graph is connected, undirected and simple, i.e., the edges are unordered pairs between distinct vertices. General vertices will be denoted by $v \in V$ unless otherwise named. Edges will be denoted by $e \in E$ if their endpoints are not required in the equations and $\{v_1, v_2\} \in E$ when the endpoints are used. Given a vector $\vec{x} \in \mathbb{R}^{|V|}$ with $|V|=n$ real values at the vertices, we denote by $x_v \in \mathbb{R}$ the scalar coefficient in $\vec{x}$ associated with vertex $v \in V$. 

We call a subgraph $T$ a spanning forest of $G$ if the vertex set of $T$ is $V$ and $T$ contains a subset of the edges $E$ such that $T$ does not contain any cycles. If $T$ is furthermore connected, we call $T$ a spanning tree. We denote the set of all trees in $G$ by $\mathcal{T}(G)$, which is non-empty due to the assumption that $G$ is connected.

When we consider edge-weighted graphs, we denote the weights by $\{w(e)\}_{e\in E}$ with positive weights $w(e) > 0$. In the case that a weight $w(e)$ might be zero, we restrict the problem to a subgraph with zero-weighted edges removed.

For images, we consider $V$ to be the set of all pixels and $E$ to denote the connection between neighbouring pixels.

\subsection{Markov Random Fields}
Formally, a Markov random field on a graph $G = (V, E)$ is a probability distribution on $\vec{x}$ with probability density such that for any two vertices $v_1, v_2 \in V$ such that $\{v_1, v_2\} \not\in E$, the coefficients $x_{v_1}$ and $x_{v_2}$ are independent when conditioned on all other coefficients of $\vec{x}$. Within this work, we use the more restrictive construction that is also called a difference prior.

\begin{definition}[Markov Random Field]\label{def:MRF}
    Let $G = (V, E)$ be a graph with edge-weights $\{w(e)\}_{e\in E}$, pick a root $r \in V$ and a positive root weight $\lambda > 0$. Let $\phi$ be the probability density of a zero-location, unimodal, unit-scale probability distribution. We define a Markov random field (MRF) on $G$ rooted at $r$ with difference densities $\phi$ by the proper probability density:
    \begin{equation}\label{eq:MRF_def}
        \pi(\vec{x}) \propto \phi\left(\lambda x_r\right)\prod_{\{v_1, v_2\} \in E} \phi\left(w(\{v_1, v_2\})(x_{v_1} - x_{v_2})\right).
    \end{equation}

    If unrooted, we consider the improper density:
    \begin{equation}\label{eq:MRF_def_unrooted}
        \pi(\vec{x}) \propto \prod_{\{v_1, v_2\} \in E} \phi\left(w(\{v_1, v_2\})(x_{v_1} - x_{v_2})\right).
    \end{equation}
    
    The important special cases in this article are:
    \begin{itemize}
        \item [\textbf{GMRF}]: $\phi(z) = (2\pi)^{-\frac{1}{2}}\exp\left(-\frac12 z^2\right)$ is the density of a standard Normal distribution.
        \item [\textbf{LMRF}]: $\phi(z) = \frac12\exp\left(-|z|\right)$ is the density of a standard Laplace distribution.
        \item [\textbf{CMRF}]: $\phi(z) = (\pi(1+z^2))^{-1}$ is the density of a standard Cauchy distribution.
    \end{itemize}
\end{definition}

\subsubsection*{Gaussian Markov Random Fields}
Of particular importance, both for modeling and computationally, is the GMRF. Given a graph $G = (V, E)$ with edge weights $\{w(e)\}_{e\in E}$ and assume an arbitrary direction of edges $\{v_1, v_2\} \in E$. We can assume this because differences will always occur as $(x_{v_1} - x_{v_2})^2$, see Definition \ref{def:MRF}. Then we define a finite difference matrix $D \in \mathbb{R}^{|E|\times|V|}$ by
\begin{equation*}
    D_{\{v_1, v_2\}, v} = \begin{cases}
    1, &\text{ if } v = v_1,\\
    -1, &\text{ if } v = v_2, \text{ and}\\
    0, &\text{ otherwise}.\\
    \end{cases}
\end{equation*}

Let $W = \text{diag}(\{w(e)\}_{e\in E}) \in \mathbb{R}_{> 0}^{|E|\times|E|}$ be a diagonal matrix of the positive edge weights. Then, the improper GMRF is given by the improper density
\begin{equation}\label{eq:implicit_GMRF}
    \pi(\vec{x}) \,\propto\, \exp\left( -\frac{1}{2}\vec{x}^\top D^\top W^2 D\vec{x}\right),
\end{equation}
where the associated precision matrix is known as the weighted graph Laplacian $L = D^TW^2D \in \mathbb{R}^{|V|\times|V|}$. It is straightforward to see that \eqref{eq:implicit_GMRF} cannot be normalized due to the all-ones vector lying in the null space of $D$, hence also in that of $L$. This is fixed by considering a proper GMRF rooted at a vertex $r$, which has proper density 
\begin{equation}\label{eq:rooted_GMRF}
    \pi(\vec{x}) \,\propto\, \exp\left( -\frac{1}{2}\vec{x}^\top(D^\top W^2D + \lambda^2 \vec{e}_r\vec{e}_r^\top)\vec{x}\right),
\end{equation}
where $e_v$ denotes the vector that is $1$ at vertex $v$ and $0$ elsewhere.

To incorporate hyperpriors on the edge weights $w$ in \eqref{eq:implicit_GMRF} or \eqref{eq:rooted_GMRF}, we need to know the proportionality not just with respect to $\vec{x}$, but also with respect to the edge weights $w$. It is common, e.g., in the horseshoe prior \cite{uribe2023horseshoe}, to assume independence of the random edge-weights, such that when conditioning on $\vec{x}$, each edge weight can be considered independently and the problem of computing the proportionality reduces to computing the normalization constant of a one-dimensional Gaussian. If we want to incorporate some dependence between edge weights $w$, we need to compute the full proportionality constant and its dependence on all edge weights $w$. Hence we need to know the determinant of the rooted weighted graph Laplacian. For this matrix and the precision matrix of \eqref{eq:rooted_GMRF}, the following classic result holds.

\begin{theorem}[The matrix-tree theorem, weighted version, see e.g. \cite{lyons2017probability}]\label{thm:matrix_tree_theorem}
Let $G = (V, E)$ be any graph with positive edge weights $\{w(e)\}_{e\in E}$ and denote by $L_v$ the matrix $L$ with the row and column associated with vertex $v\in V$ removed. Then for any $v \in V$,
    \begin{equation}\label{eq:matrix_tree_theorem}
        \det(L_v) = \sum_{T \in \mathcal{T}(G)} \prod_{e \in T}w(e)^2.
    \end{equation}

    Furthermore, when including a root $r\in V$ it holds that
    \begin{equation}\label{eq:matrix_tree_theorem_ext}
        \det(L + \lambda^2 e_re_r^T) = \lambda^2\sum_{T \in \mathcal{T}(G)} \prod_{e \in T}w(e)^2.
    \end{equation}
\end{theorem}

\begin{corollary}[GMRF normalization constant]\label{cor:GMRF_density}
    The density of the rooted GMRF on a graph $G = (V, E)$ is
    \begin{equation}\label{eq:GMRF_proportionality}
        \pi(\vec{x}) = (2\pi)^{-\frac{|V|}{2}}\lambda \left(\sum_{T \in \mathcal{T}(G)} \prod_{e \in T}w(e)^2\right)^{\frac{1}{2}}\exp\left( -\frac{1}{2}\vec{x}^T(D^T W^2D + \lambda^2 \vec{e}_r\vec{e}_r^T)\vec{x}\right).
    \end{equation}

    If $G$ is itself a tree, then this density simplifies to
    \begin{equation}\label{eq:GMRF_proportionality_tree}
        \pi(\vec{x}) = (2\pi)^{-\frac{|V|}{2}}\lambda \prod_{e \in G}w(e)\exp\left( -\frac{1}{2}\vec{x}^T(D^T W^2D + \lambda^2 \vec{e}_r\vec{e}_r^T)\vec{x}\right).
    \end{equation}
\end{corollary}

\begin{figure}
    \centering
    \includegraphics[width=1.0\linewidth]{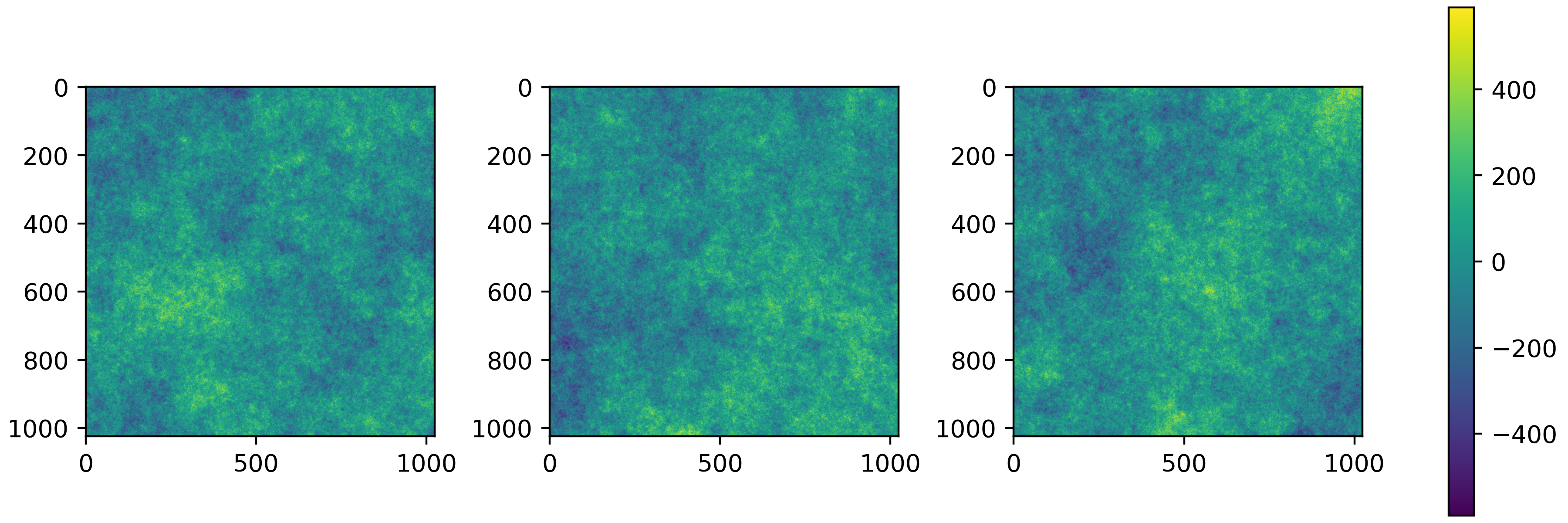}
    \caption{Samples of a GMRF with constant weights.}
    \label{fig:GMRF_samples}
\end{figure}

\subsubsection*{Laplace and Cauchy Markov Random Fields}
For the LMRF and CMRF, computing the normalization constant of a graph with cycles, like for \eqref{eq:GMRF_proportionality}, is intractable. However, for a rooted tree the densities admit a similar closed-form factorization to \eqref{eq:GMRF_proportionality_tree}.
\begin{equation}\label{eq:LMRF_proportionality_tree}
        \pi(\vec{x}) = \left(\frac12\right)^{|V|}\lambda\exp\left(-\lambda |x_r|\right)\prod_{\{v_1, v_2\} \in E} w(\{v_1, v_2\})\exp\left(-w(\{v_1, v_2\})|x_{v_1} - x_{v_2}|\right),
\end{equation}
for the LMRF and 
\begin{equation}\label{eq:CMRF_proportionality_tree}
        \pi(\vec{x}) = \left(\frac{1}{\pi}\right)^{|V|}\lambda \left( 1 + \lambda^2x_r^2\right)^{-1}\prod_{\{v_1, v_2\} \in E} w(\{v_1, v_2\})\left(1 + w(\{v_1, v_2\})^2(x_{v_1} - x_{v_2})^2\right)^{-1},
\end{equation}
for the CMRF. These normalization structures naturally suggest describing a prior structure in terms of random spanning trees, which we review in Subsection \ref{subsec:RST}.

If the prior is Gaussian, like a GMRF, the forward operator is linear, and the additive error is Gaussian, then the posterior distribution is a Gaussian distribution, for which many efficient sampling algorithms exist, see \cite{vono2022high}. However, this does not hold for the LMRF and CMRF priors. For these priors, the prior can be transformed into a hierarchical prior by introducing auxiliary random variables, such that the prior is Gaussian conditioned on these newly introduced random variables. Concretely, let $z$ represent the random difference between two neighbouring pixels, then we can use the following scale-mixture representations with scale $\gamma$, see \cite{flock2025continuous, senchukova2024bayesian, gelman2006prior}:
\begin{align}
    z \sim \text{Laplace}(0, \gamma) &\text{ is equivalent to } z|\tau \sim \mathcal{N}(0, \tau) \text{ and } \tau \sim \text{Exp}\left(\frac{\gamma^2}{2}\right), \text{ and} \label{eq:scale-mixture_laplace}\\
    z \sim \text{Cauchy}(0, \gamma) &\text{ is equivalent to } z|\tau\sim \mathcal{N}(0, \tau) \text{ and } \tau \sim \text{IG}\left(\frac{1}{2}, \frac{\gamma^2}{2}\right). \label{eq:scale-mixture_cauchy}
\end{align}

For efficient computation by Gibbs sampling, we will also make use of the converse conditionals, i.e., $\tau \,|\, z$. For the Laplace scale-mixture representation, this gives the following relation \cite{park2008bayesian}:
\begin{equation}
    \tau|z \sim \text{InvGaussian}\left(\frac{\gamma}{|z|}, \gamma^2\right),
\end{equation}
and for the Cauchy scale-mixture representation \cite{senchukova2024bayesian}:
\begin{equation}
    \tau^{-1}|z \sim \text{Gamma}\left(1, \frac{z^2 + \gamma^2}{2}\right).
\end{equation}

\subsection{Random Spanning Trees}\label{subsec:RST}
In this subsection, we consider the problem of sampling random subgraphs from a graph $G = (V, E)$. In the simplest model each edge is removed independently with some probability, which is equivalent to the Bernoulli/bond percolation model studied in percolation theory \cite{shante1971introduction}. To introduce correlation between edges, we consider random spanning trees, a few samples of which are shown in Figure \ref{fig:RST_examples} and an extensive study can be found in \cite{lyons2017probability}.

\begin{definition}[weighted random spanning tree (WRST)]\label{def:RST}
    Given a graph $G=(V, E)$ with edge weights $\{w(e)\}_{e\in E}$, we define a random spanning tree $T$ on $G$ by the probability
    \begin{equation}\label{eq:weights_tree}
        \mathbb{P}(T = \tilde{T}) \, \propto \, \prod_{e\in \tilde{T}} w(e),
    \end{equation}
    where the normalization constant is given by $\det(L_v)$ for $L_v = D^TWD$ with the row and column  corresponding to any $v\in V$ removed, by the matrix-tree Theorem \ref{thm:matrix_tree_theorem}.
\end{definition}

\begin{figure}
    \centering
    
    \begin{subfigure}{1.0\textwidth}
        \hspace{5em}
        \begin{tikzpicture}[scale=1.2,
    vertex/.style={circle, draw, fill=white, inner sep=2pt},
    gridedge/.style={gray!40, line width=0.6pt},
    treeedge/.style={black, line width=1.4pt}
]

% Vertices
\foreach \x in {0,1,2} {
  \foreach \y in {0,1,2} {
    \node[vertex] (\x\y) at (\x,\y) {};
  }
}

% Full grid edges (background)
\foreach \x in {0,1,2} {
  \foreach \y in {0,1} {
    \draw[gridedge] (\x\y) -- (\x\the\numexpr\y+1\relax);
    \draw[gridedge] (\y\x) -- (\the\numexpr\y+1\relax\x);
  }
}

% Random spanning tree edges (foreground)
\draw[treeedge] (00) -- (01);
\draw[treeedge] (01) -- (02);
\draw[treeedge] (02) -- (12);
\draw[treeedge] (12) -- (11);
\draw[treeedge] (11) -- (10);
\draw[treeedge] (10) -- (20);
\draw[treeedge] (20) -- (21);
\draw[treeedge] (21) -- (22);

\end{tikzpicture}\hfill
        \begin{tikzpicture}[scale=1.2,
    vertex/.style={circle, draw, fill=white, inner sep=2pt},
    gridedge/.style={gray!40, line width=0.6pt},
    treeedge/.style={black, line width=1.4pt}
]

% Vertices
\foreach \x in {0,1,2} {
  \foreach \y in {0,1,2} {
    \node[vertex] (\x\y) at (\x,\y) {};
  }
}

% Full grid edges (background)
\foreach \x in {0,1,2} {
  \foreach \y in {0,1} {
    \draw[gridedge] (\x\y) -- (\x\the\numexpr\y+1\relax);
    \draw[gridedge] (\y\x) -- (\the\numexpr\y+1\relax\x);
  }
}

% Random spanning tree edges (foreground)
\draw[treeedge] (00) -- (01);
\draw[treeedge] (01) -- (02);
\draw[treeedge] (20) -- (21);
\draw[treeedge] (21) -- (22);
\draw[treeedge] (00) -- (10);
\draw[treeedge] (10) -- (20);
\draw[treeedge] (12) -- (22);
\draw[treeedge] (12) -- (11);

\end{tikzpicture}\hfill
        \begin{tikzpicture}[scale=1.2,
    vertex/.style={circle, draw, fill=white, inner sep=2pt},
    gridedge/.style={gray!40, line width=0.6pt},
    treeedge/.style={black, line width=1.4pt}
]

% Vertices
\foreach \x in {0,1,2} {
  \foreach \y in {0,1,2} {
    \node[vertex] (\x\y) at (\x,\y) {};
  }
}

% Full grid edges (background)
\foreach \x in {0,1,2} {
  \foreach \y in {0,1} {
    \draw[gridedge] (\x\y) -- (\x\the\numexpr\y+1\relax);
    \draw[gridedge] (\y\x) -- (\the\numexpr\y+1\relax\x);
  }
}

% Random spanning tree edges (foreground)
\draw[treeedge] (00) -- (10);
\draw[treeedge] (10) -- (20);
\draw[treeedge] (01) -- (11);
\draw[treeedge] (11) -- (21);
\draw[treeedge] (02) -- (12);
\draw[treeedge] (12) -- (22);
\draw[treeedge] (10) -- (11);
\draw[treeedge] (11) -- (12);

\end{tikzpicture}
        \hspace{5em}
        \caption{Three samples of a random spanning tree on a 3-by-3 grid.}
        \label{fig:RST_examples:smoll}
    \end{subfigure}
    
    \begin{subfigure}{1.0\textwidth}
    
    \centering
        \includegraphics[width=0.8\textwidth]{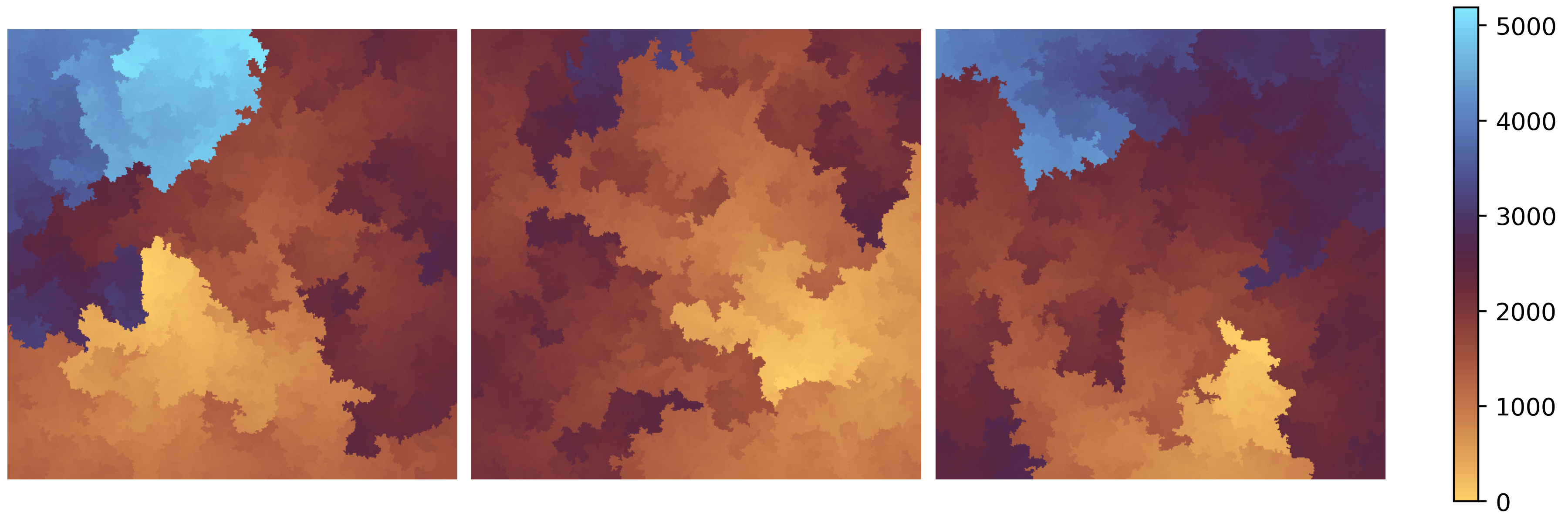}
        \caption{Three samples of a random spanning tree on a 512-by-512 grid. The color represents the graph distance to a randomly chosen root vertex.}
        \label{fig:RST_examples:large}
    \end{subfigure}
    
    \caption{Examples of small and large scale random spanning trees.}
    \label{fig:RST_examples}
\end{figure}
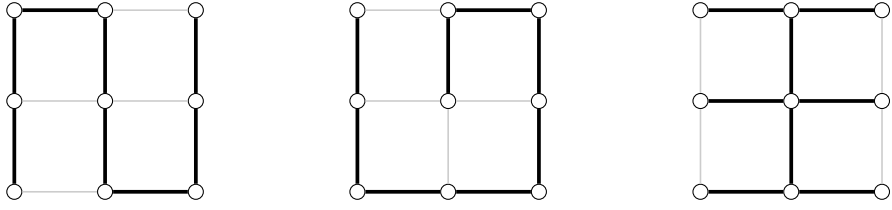
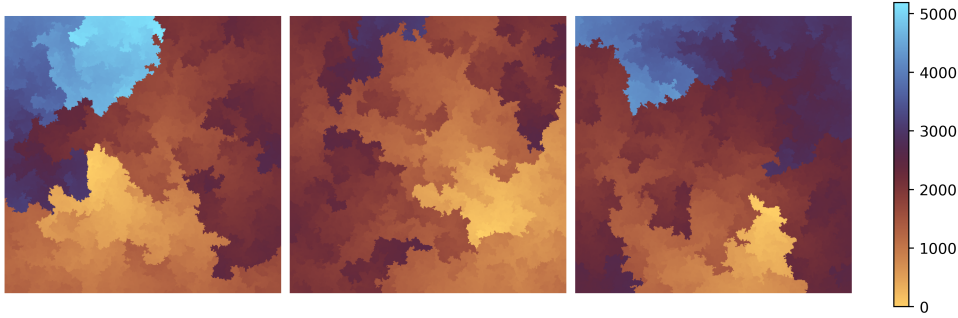

\subsubsection*{Sampling from random spanning trees}
Finding the tree of maximum probability is a well-studied problem in computer science, with deterministic algorithms such as Prim's algorithm \cite{prim1957shortest} and Kruskal's algorithm \cite{kruskal1956shortest} having a runtime of $\mathcal{O}(|E|\log|V|)$. Sampling from random spanning trees has found popular use in generating random mazes. Classical algorithms for exact sampling from weighted random spanning trees, as defined in Definition \ref{def:RST}, are the Aldous-Broder algorithm \cite{aldous1990random, broder1989generating} and the generally faster Wilson's algorithm, see \cite{wilson1996generating, propp1998get} or a modern discussion in \cite[Section 4.1]{lyons2017probability}.

Wilson's method is a stochastic algorithm based on loop-erased random walks. We consider a random walk on $G$ where the probability of going from one vertex $v_o$ to its neighbor $v_n$ is proportional to the weight $w(\{v_o, v_n\})$ of the edge connecting $v_o$ to $v_n$. We keep track of the trajectory of the random walk and when the random walk revisits a vertex on its own trajectory, the newly introduced loop is erased, hence a loop-erased random walk. To construct the random spanning tree, we pick an arbitrary initial vertex as the root of the tree. Then, start from any vertex that has yet to be included in the tree and perform a loop-erased random walk until the tree is hit. Then add the loop-erased trajectory to the tree and start a walk from a new vertex. Repeating this until all vertices have been added results in a random tree according to the weights \eqref{eq:weights_tree}.

Due to the stochastic nature of the loop erasure, the algorithm does not have a fixed runtime upper bound. However it is known that the expected runtime is $\mathcal{O}(\text{cover time of }G)$, where the cover time is the expected number of steps a random walk needs to visit every vertex in the graph. Wilson's algorithm is relatively efficient for uniform spanning trees, i.e., where all edge weights are equal, with an expected runtime bounded by $\mathcal{O}(|E||V|)$ for general graphs \cite{aleliunas1979random} and $\mathcal{O}(|V|\log|V|)$ for a 2D grid \cite{aldous1989lower, aldous-fill-2014}, i.e., an image. For non-uniform random spanning trees, it is possible that two regions of the graph are only connected by edges with relatively low weights, causing the random walk to only jump between regions after a very long time, see Figure \ref{fig:RST_bottleneck_example}. This can increase the runtime by a factor $\frac{w_{\max}}{w_{\min}}$, where $w_{\max}$ and $w_{\min}$ are the largest and smallest edge weight respectively. This would make Wilson's algorithm rather slow in these cases. Whilst recent progress has been made on circumventing such bottlenecks, see \cite{tam2025exact}, we propose in Section \ref{sec:method} to modify the graph to better control the runtime.

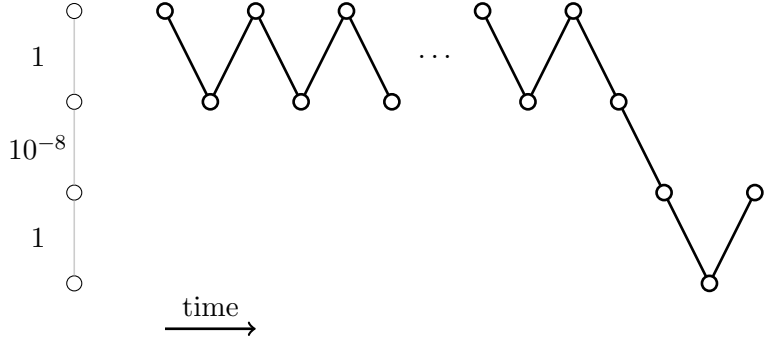
\begin{figure}
    \centering
    \hspace{5em}
    \begin{tikzpicture}[scale=1.2,
    vertex/.style={circle, draw, fill=white, inner sep=2pt},
    gridedge/.style={gray!40, line width=0.6pt},
    treeedge/.style={black, line width=1.0pt}
]

% Number of points
\def\N{4}

% Vertices
\foreach \x in {0,...,\numexpr\N-1\relax} {
  \node[vertex] (\x) at (0,\x) {};
}

% Edges
\foreach \x in {0,...,\numexpr\N-2\relax} {
  \draw[gridedge] (\x) -- (\the\numexpr\x+1\relax);
}

\node[] at (-0.4,0.5) {$1$};
\node[] at (-0.4,1.5) {$10^{-8}$};
\node[] at (-0.4,2.5) {$1$};

\draw[treeedge, ->] (1, -0.5) -> (2, -0.5);
\node[] at (1.5, -0.25) {time};

\draw[treeedge] (1, 3) node[vertex] {}
            -- (1.5, 2) node[vertex] {}
            -- (2, 3) node[vertex] {}
            -- (2.5, 2) node[vertex] {}
            -- (3, 3) node[vertex] {}
            -- (3.5, 2) node[vertex] {};
\node[] at (4,2.5) {\dots};
\draw[treeedge] (4.5, 3) node[vertex] {}
                -- (5, 2) node[vertex] {}
                -- (5.5, 3) node[vertex] {}
                -- (6, 2) node[vertex] {}
                -- (6.5, 1) node[vertex] {}
                -- (7, 0) node[vertex] {}
                -- (7.5, 1) node[vertex] {};

\end{tikzpicture}
    \hspace{5em}
    \caption{Example of a random walk on a weighted graph with a bottleneck.}
    \label{fig:RST_bottleneck_example}
\end{figure}

\section{Method}\label{sec:method}

In this section, we combine the random spanning trees with the Markov random field priors, resulting in the random spanning tree Markov random fields (RST-MRFs). We derive a conjugacy relation and an efficient Gibbs sampler for exploring the posterior distribution. We also discuss why fractal-like interfaces appear in high-resolution samples of the prior and how to improve the efficiency of weighted random spanning tree sampling in the case of bottlenecks.

\subsection{Random spanning tree hyperprior}
For the RST-MRF prior, we construct a hierarchical prior by combining the MRF from \eqref{eq:MRF_def} with the weighted random spanning tree model from \eqref{def:RST}. We assume all the edges and the root in the MRF have the same scale parameter $\lambda > 0$. Then, the resulting joint prior over the continuous variable $\vec{x}$ and discrete tree $T$ takes the form:
\begin{align}
    \pi(\vec{x}\,|\, T) \, &\propto\,  \phi\left(\lambda x_r\right)\prod_{\{v_1, v_2\} \in T} \phi\left(\lambda (x_{v_1} - x_{v_2})\right), \label{eq_RST-MRF}\\
    \mathbb{P}(T = \tilde{T}) \, &\propto \, \prod_{e\in \tilde{T}} w(e). \nonumber
\end{align}
Without knowing anything about the connectivity in the image, we assume a uniform random spanning tree hyperprior, i.e., $w(e) = 1$ for all $e \in E$.

Note that in the Gaussian case, the RST-GMRF is a Gaussian mixture model, combining one Gaussian distribution for each spanning tree in the graph. In the rooted case, each conditional $\vec{x}\,|\, T$ is a proper probability distribution, hence the prior is proper. In the unrooted case, the negative log-densities $-\log(\pi(\vec{x}\,|\, T))$ are coercive on the subspace orthogonal to $\vec{1}$ for every tree $T$. Hence, under standard Gaussian likelihood assumptions, the posterior will be proper if and only if $\vec{1} \not\in \text{Null}(A)$, just like for ordinary GMRFs, LMRFs and CMRFs.

\begin{figure}
    \centering
    \includegraphics[width=1.0\linewidth]{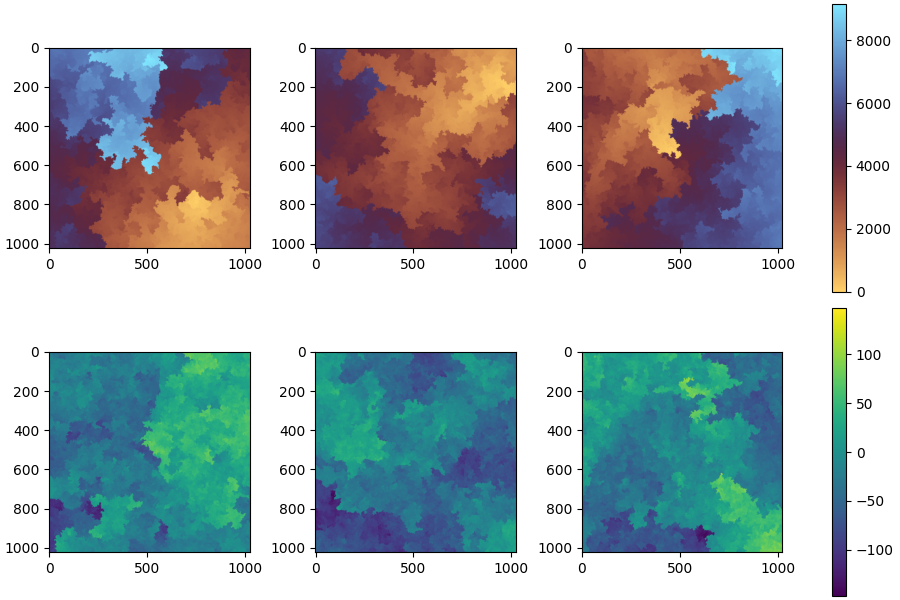}
    \caption{Three samples of uniform RST-GMRF, with the random spanning tree (top) and image sample (below.)}
    \label{fig:RST-GMRF_samples}
\end{figure}

The key to efficient sampling is the observation that weighted random spanning trees are conjugate prior distributions with respect to MRFs as defined in Definition \ref{def:MRF}. This means that $T\,|\, \vec{x}$ is itself also a weighted random spanning tree, making Gibbs sampling efficient.

\begin{theorem}[Random spanning tree hyperprior conjugacy]\label{thm:RST-conjugacy}
Let $G=(V,E)$ be a connected graph and let $\phi$ be a symmetric, zero-location, unimodal around $0$, unit-scale probability density. Consider the hierarchical RST-MRF prior defined by
\begin{align}
    \pi(\vec{x}\mid T)
    &\propto
    \phi(\lambda x_r)
    \prod_{\{v_1,v_2\}\in T}
    \phi\!\left(\lambda(x_{v_1}-x_{v_2})\right), \label{eq:thm_1}\\
    \mathbb{P}(T=\tilde T)
    &\propto
    \prod_{e\in\tilde T} w(e), \label{eq:thm_2}
\end{align}
with $\lambda>0$ and positive edge weights $\{w(e)\}_{e\in E}$.

Then, for fixed $\vec{x}$, the conditional distribution of the tree $T\mid\vec{x}$
is again a weighted random spanning tree on $G$, with edge weights
\begin{equation}
    \widetilde w(\{v_1,v_2\})
    \;=\;
    w(\{v_1,v_2\})\,
    \phi\!\left(\lambda(x_{v_1}-x_{v_2})\right).
\end{equation}
\end{theorem}

\begin{proof}
Combining \eqref{eq:thm_1} and \eqref{eq:thm_2} using Bayes' rule gives, for any $\tilde{T}\in\mathcal{T}(G)$,
\begin{align*}
\mathbb{P}(T=\tilde{T}\mid \vec{x})
&\propto \mathbb{P}(T=\tilde{T})\,\pi(\vec{x}\mid \tilde{T})\\
&\propto \left(\prod_{e\in\tilde{T}} w(e)\right)\,
          \phi(\lambda x_r)\prod_{\{v_1,v_2\}\in\tilde{T}}
          \phi\!\left(\lambda(x_{v_1}-x_{v_2})\right).
\end{align*}
The root term $\phi(\lambda x_r)$ does not depend on $\tilde{T}$ and cancels in the
proportionality. Therefore,
\begin{equation*}
\mathbb{P}(T=\tilde{T}\mid \vec{x})
\ \propto\ \prod_{\{v_1,v_2\}\in\tilde{T}}
\underbrace{\Bigl[w(\{v_1,v_2\})\,\phi\!\left(\lambda(x_{v_1}-x_{v_2})\right)\Bigr]}_{\widetilde{w}(\{v_1,v_2\})}.
\end{equation*}
This is precisely the law of a weighted random spanning tree with updated edge weights
$\widetilde{w}(\{v_1,v_2\})=w(\{v_1,v_2\})\,\phi(\lambda(x_{v_1}-x_{v_2}))$.
\end{proof}

In Theorem \ref{thm:RST-conjugacy}, the assumptions that $\phi$ is a zero-location, unimodal, unit-scale probability density are not used, but they have important consequences. Unit-scale highlights that we can tune the global scale parameter $\lambda$ ourselves. The assumptions that $\phi$ is symmetric and unimodal around $0$ imply that a larger magnitude jump $|x_{v_1}-x_{v_2}|$ will decrease the conjugate weights $w(\{v_1,v_2\})\,\phi(\lambda(x_{v_1}-x_{v_2}))$. Therefore, trees $T$ containing larger magnitude jumps have lower conditional probability $\mathbb{P}(T=\tilde{T}\mid \vec{x})$ than trees that contain only small magnitude jumps.

\subsection{Fractal-like interfaces}
Figure \ref{fig:RST-GMRF_samples} shows a few samples of random spanning trees on $1024$-by-$1024$ images and associated GMRF samples. These samples show the rich fractal-like structure of random spanning trees and how they transfer into RST-GMRF samples. We can show that when, the image size becomes larger, the fractal-like interfaces do indeed become fractals. In particular, look at Figure \ref{fig:RST-GMRF_sample_zoom_levels}, which contains a large dimensional sample at various zoom levels. At the closest zoom level, the image is split into two colors. On that scale, the graph distance between any pixels of the same color is very small compared to the graph distance between the two colors. By taking this to the extreme using a simplified model illustrated in Figure \ref{fig:RST-proof}, we can show that in the limit these interfaces converge to fractals.

\begin{proposition}\label{thm:fractal_interface}
    Consider a square domain $[0,1]^2$ with square-lattice discretization $\delta \mathbb{Z} \cap [0,1]^2$ and color the boundary of the left half blue and the right half red. Just like Wilson's algorithm, perform loop-erased random walks until the boundary is hit and color the trajectory based on which boundary it connects to. This is equivalent to sampling from a random spanning tree on the graph with the boundaries connected to some common terminal node as visualized in Figure \ref{fig:RST-proof}. Consider the random curve describing the interface between the red and blue from the left to right side of the square. As the discretization becomes finer, i.e., $\delta$ goes to zero, this random curve converges to a chordal Schramm-Loewner evolution (SLE\textsubscript{$\kappa$}) with $\kappa = 2$, which almost surely has fractal dimension $1.25$.
\end{proposition}

The proof of this proposition relies on framing the random interface as a chordal Schramm-Loewner evolution, an important tool in studying limits of lattice models in statistical mechanics, and rooted in stochastic and complex analysis. The main conclusion follows directly from standard properties of the Schramm-Loewner evolution. Providing a proper definition and discussion of the SLE\textsubscript{$\kappa$} is beyond the scope of this article, hence we quickly sketch the relation between the random spanning tree and the random interface.

\begin{proof}[Sketch of proof]
    Observe that the graph is planar, hence we can consider a dual graph. Dual to the sampled uniform random spanning tree on the graph, there exists a dual uniform random spanning tree on the dual graph of the square-lattice discretization. In this dual spanning tree, there exists a unique path from the top of the square to the bottom of the square; this path represents the interface between the red and blue regions. This unique path on the dual random spanning tree is distributed according to a loop-erased random walk between the start and end point. The path of a loop-erased random walk on our discretization converges to SLE\textsubscript{$\kappa$} with $\kappa = 2$ with associated fractal dimension is $1 + \frac{\kappa}{8} = 1.25$, see \cite{beffara2008dimension, lawler2011conformal}. 
\end{proof}

\begin{figure}
    \centering
    \includegraphics[width=1.0\linewidth]{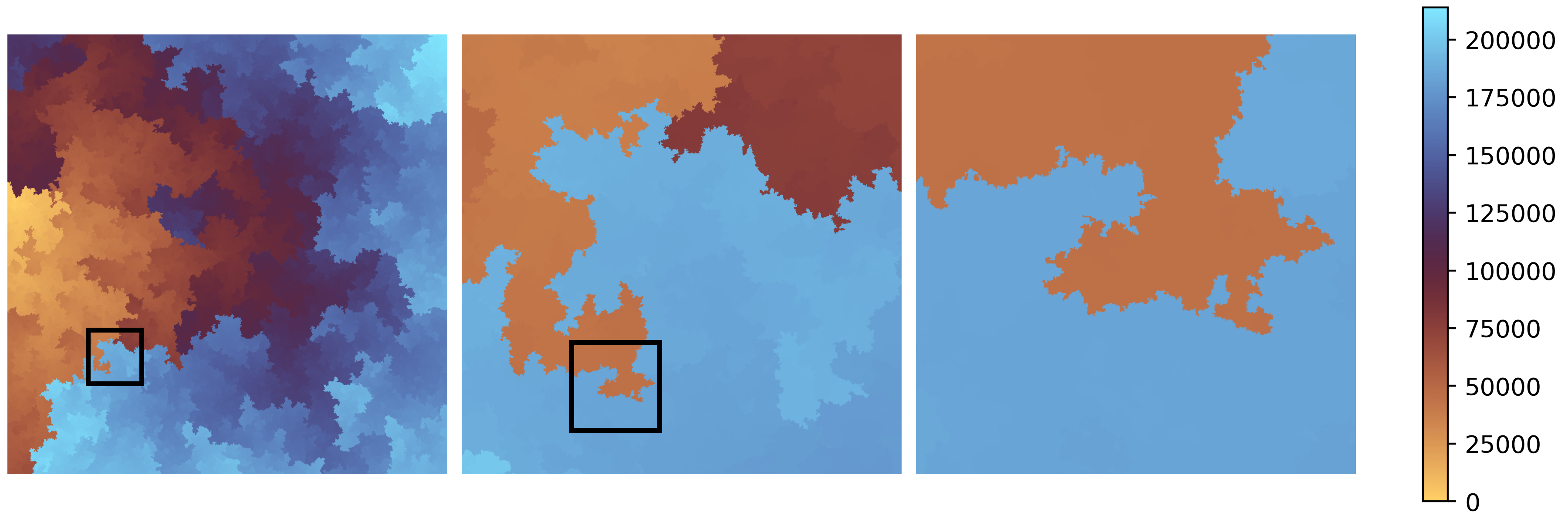}\hfill
    \includegraphics[width=1.0\linewidth]{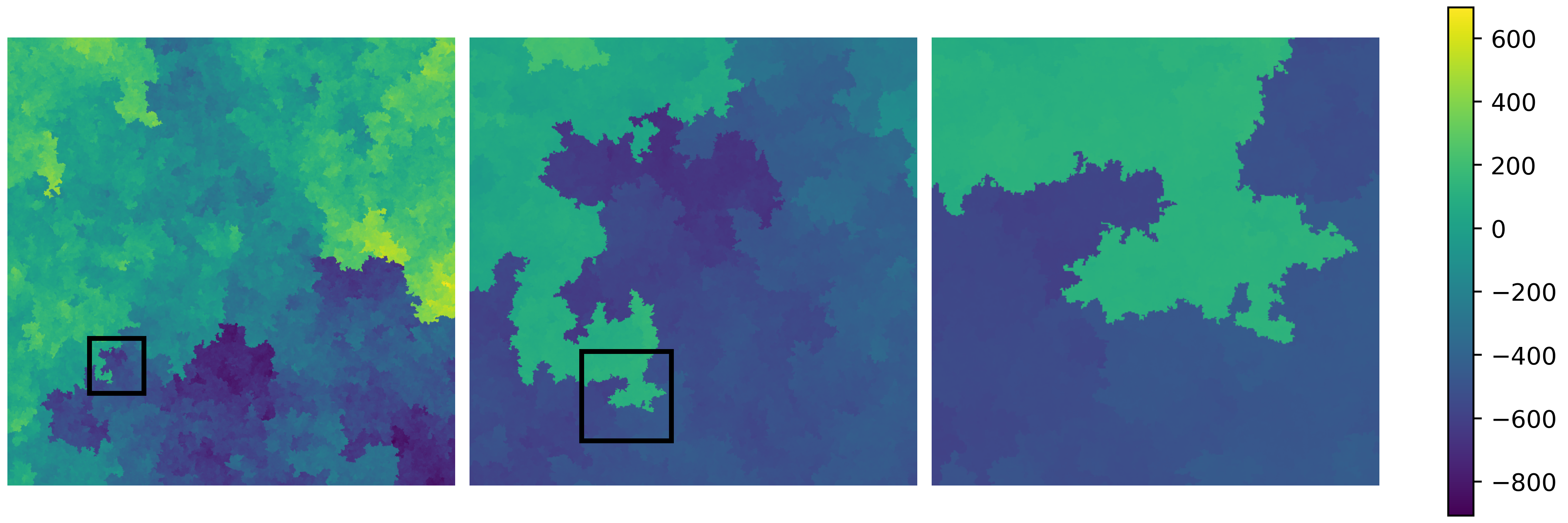}\hfill
    \caption{Example of an RST-GMRF sample at various zoom levels. The left most image is size $16384^2$, with the right-most zoom being $400^2$.}
    \label{fig:RST-GMRF_sample_zoom_levels}
\end{figure}

\begin{figure}
    \centering
    \input{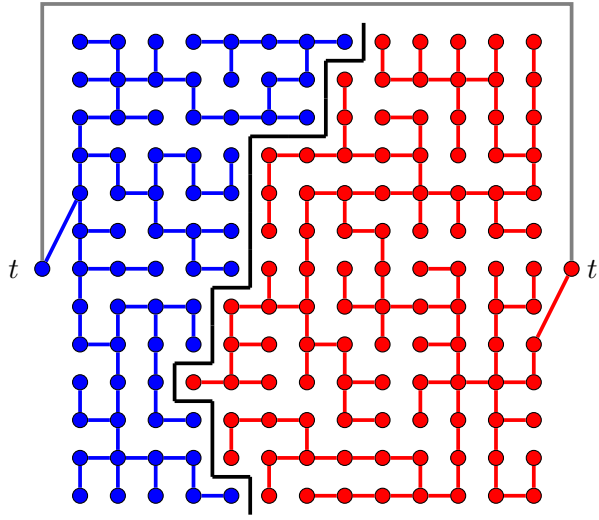}
    \caption{Example of an interface separating two components of a spanning tree.}
    \label{fig:RST-proof}
\end{figure}

\subsection{Algorithms}
Sampling directly from the posterior distribution on $(\vec{x}, T)$ under the RST-MRF prior is expected to be rather difficult due to the coupling of the continuous image variable $\vec{x}$ and the discrete spanning tree $T$. Hence, we consider a Gibbs sampler that alternates between $\vec{x}\,|\,T, \vec{y}$ and $T\,|\,\vec{x}$, equivalently $T\,|\,\vec{x}, \vec{y}$. When conditioning on the spanning tree $T$, the resulting distribution is equivalent to a posterior distribution with a tree-structured MRF prior, from which we can sample efficiently using a linear randomize-then-optimize approach combined with scale-mixture representations for LMRF \eqref{eq:scale-mixture_laplace} and CMRF \eqref{eq:scale-mixture_cauchy}. When conditioning on the image $\vec{x}$, we use Theorem \ref{thm:RST-conjugacy} to conclude that this distribution is a reweighted random spanning tree, from which we can sample using Wilson's algorithm. This sampling scheme is summarized in Algorithm \ref{alg:gibbs_rst-mrf}.

\begin{algorithm}
\caption{Gibbs sampler for linear inverse problems with RST-MRF priors}
\label{alg:gibbs_rst-mrf}
\begin{algorithmic}[1]

\STATE \textbf{Input:} data $y$, forward operators $A$, noise variance $\sigma^2$, iterations $N$
\STATE \textbf{Initialize:} $x^{(0)}$

\FOR{$t = 1$ to $N$}
    \STATE Compute conjugated graph-weights $w^{(t)}$ from $x^{(t-1)}$ using Theorem \ref{thm:RST-conjugacy}
    \STATE Sample random spanning tree $T^{(t)}$ with weights  $w^{(t)}$, e.g., using Wilson's algorithm
    \STATE Compute a finite difference matrix $D$ corresponding to the spanning tree $T^{(t)}$
    \STATE Sample $x^{(t)}$ from $\pi(\vec{x}\,|\, \vec{y}, D)$, e.g., using a scale-mixture representation.
\ENDFOR

\STATE \textbf{Output:} $\{(x^{(t)}, T^{(t)})\}_{t=1}^N$

\end{algorithmic}
\end{algorithm}

Due to the conjugacy in Theorem \ref{thm:RST-conjugacy}, large magnitude jumps $|x_{v_1}-x_{v_2}|$ in the image or a large prior strength $\lambda$ can cause large weight ratios in $T\,|\,\vec{x}$, which can greatly increase the runtime of sampling algorithms as discussed at the end of Subsection \ref{subsec:RST}. These large weight ratios can cause bottlenecks for random-walk-based spanning tree sampling algorithms, which can get stuck in regions separated by low-weight edges, e.g., as illustrated in Figure \ref{fig:RST_bottleneck_example}. One approach to circumvent this issue is to enforce a lower bound on the conjugate edge weights $\widetilde{w}(e)$. This makes the crossing of bottlenecks by the random walk more likely, but thereby increasing the chance of selecting edges with large jumps, which counteracts the behaviour of the hyperprior.

Instead, to prevent this computational bottleneck at the cost of a slightly different tree distribution, we can change the model by extending the graph with a single vertex. For any graph $G = (V, E)$, add a new vertex $t$ that is connected to every other vertex, i.e., we consider the new graph $G_t = (V \cup \{t\}, E \cup \{\{t, v\}\}_{v \in V})$. Keeping the existing edge weights unchanged, we let $w(\{v, t\}) = \rho$ for some termination weight $\rho > 0$. As Wilson's method can be rooted at any vertex, we will always use the new terminal vertex $t$ as root. Wilson's algorithm can now step to the new terminal vertex $t$ at any vertex with a probability depending on $\rho$. Hence, a random walk stuck in a piecewise constant region will at some point escape to the terminal vertex and bypass the low-weight bottleneck altogether. This model is illustrated in 1D in Figure \ref{fig:RST_bottleneck_example_terminal}. To be precise, the expected runtime can be bounded in terms of the number of vertices $|V|$ and the termination weight $\rho$.

\begin{figure}
    \centering
    \hspace{5em}
    \begin{tikzpicture}[scale=1.2,
    vertex/.style={circle, draw, fill=white, inner sep=2pt},
    gridedge/.style={gray!40, line width=0.6pt},
    treeedge/.style={black, line width=1.0pt}
]

% Number of points
\def\N{6}

% Vertices
\foreach \x in {0,...,\numexpr\N-1\relax} {
  \node[vertex] (\x) at (\x, 0) {};
}

% Terminal node
\node[vertex] (t) at (2.5, -1) {};
\node[] at (2.5, -1.25) {t};
  
% Edges
\foreach \x in {0,...,\numexpr\N-2\relax} {
  \draw[gridedge] (\x) -- (\the\numexpr\x+1\relax);
}
\foreach \x in {0,...,\numexpr\N-1\relax} {
  \draw[gridedge] (\x) -- (t);
}

\node[] at (0.5, 0.3) {1};
\node[] at (1.5, 0.3) {1};
\node[] at (2.5, 0.3) {$10^{-5}$};
\node[] at (3.5, 0.3) {1};
\node[] at (4.5, 0.3) {1};

\node[] at (1.0, -0.6125) {$\rho$};

\draw[treeedge] (0) -- (1) -- (2);
\draw[treeedge] (3) -- (4) -- (5);

\draw[treeedge] (1) -- (t);
\draw[treeedge] (5) -- (t);

\end{tikzpicture}
    \hspace{5em}
    \caption{Example of a random spanning tree on a weighted graph with a terminal vertex.}
    \label{fig:RST_bottleneck_example_terminal}
\end{figure}
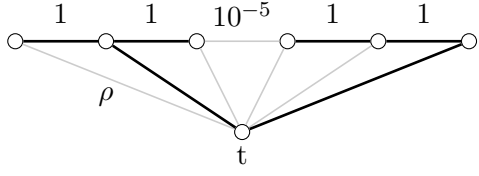

\begin{theorem}\label{thm:runtime_termination}
    Consider the extended graph as described above with termination weight $\rho > 0$ and let $d$ be the maximal degree of the original graph. The expected runtime of Wilson's algorithm with the termination vertex as root, i.e., the expected number of steps of all random walks combined, is $\mathcal{O}\left(\frac{w_{\max}d}{\rho}|V|\right)$, where $w_{\max} := \max_{e\in E}w(e)$. Furthermore, assuming a uniform random spanning tree prior, i.e., $w(e) = 1$ for all $e \in E$, and uniform prior strength, i.e., $\lambda(e) = \lambda$ for all $e \in E$, then the runtime is $\mathcal{O}\left(\frac{\phi(0)d}{\rho}|V|\right)$ when the weights are conjugated using any vector.
\end{theorem}
\begin{proof}
    At every vertex, the probability of stepping to the terminal vertex is at least $\frac{\rho}{w_{\max}d+\rho}$ due to each vertex having at most $d$ neighbors excluding the terminal vertex. Each random walk in Wilson's algorithm stops if either it hits vertices already included in the tree or the random walk jumps to the termination vertex. The number of steps until the latter event is stochastically dominated by a geometric distribution with success probability $\frac{\rho}{w_{\max}d+\rho}$, hence with expectation $\frac{w_{\max}d+\rho}{\rho} = \mathcal{O}\left(\frac{w_{\max}d}{\rho}\right)$. Given that every loop-erased random walk will always result in including the starting vertex into the tree, the total expected runtime is bounded by $\mathcal{O}\left(\frac{w_{\max}d}{\rho}|V|\right)$.

    Under the assumption of uniform weights and prior strength, the maximum weight $w_{\max}$ is achieved when difference between two neighbouring vertices is zero, in which case $w_{\max} = \phi(0)$.
\end{proof}

Compared to the runtime of Wilson's algorithm for uniform spanning trees on a 2D grid $\mathcal{O}\left(|V|\log|V|\right)$, the runtime bound in Theorem \ref{thm:runtime_termination} has no logarithmic factor. Furthermore, the possibility of very low $w_{\min}$ caused by large magnitude jumps and the rapid decay of the updated conjugate weights $\widetilde w(e)$ has been replaced by the controllable termination weight $\rho$. Whilst this extended graph model is still exact for the modified hyperprior, the newly introduced vertex and edges are not part of the unknown parameters encoded in $\vec{x} \in \mathbb{R}^{|V|}$. Unless there is only a single edge between the original graph and the terminal vertex $t$, the graph used for constructing the finite difference operator of the MRF is no longer a spanning tree, but a spanning forest. The number of connected components in this spanning forest equals the degree of the terminal vertex and increases with $\rho$. Figure \ref{fig:RST_runtime} shows how this construction can greatly improve the runtime of Wilson's algorithm, but can cause a very large number of components if not properly tuned. This approach thus requires proper tuning of $\rho$ to balance runtime and accuracy.

\begin{figure}
    \centering
    \includegraphics[width=1.0\linewidth]{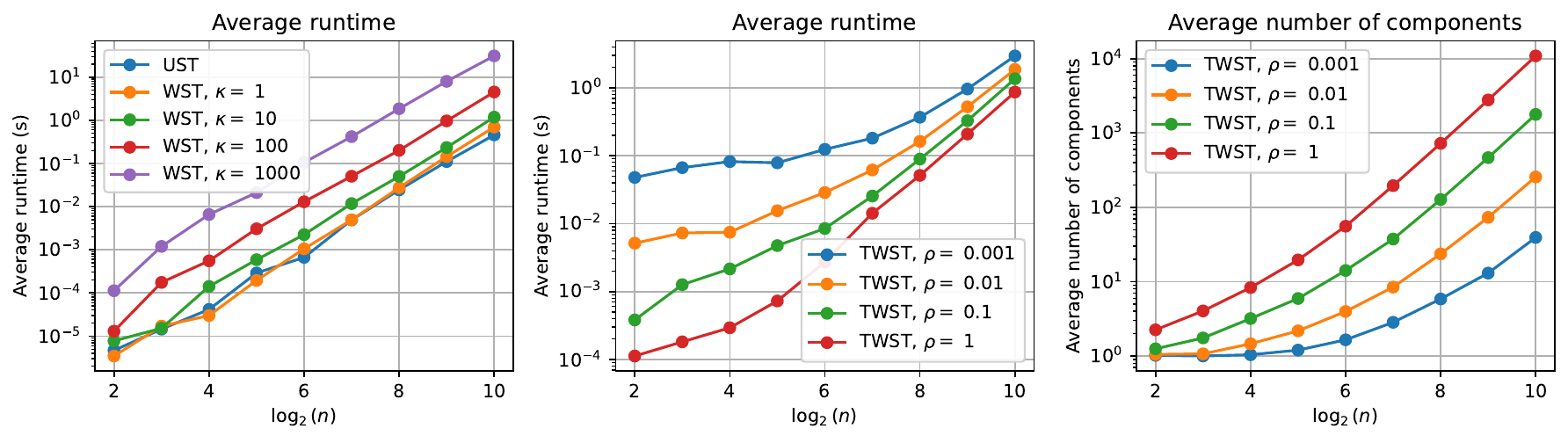}
    \caption{Average runtime over $100$ samples of uniform (UST) and weighted random spanning trees (WST) at grids of size $(n, n)$. The value $\kappa$ is the weight on the horizontal edges, when the vertical edges are fixed to one, hence creating a bottleneck for the loop-erased random walk to move over the vertical edges for large $\kappa$.}
    \label{fig:RST_runtime}
\end{figure}

\section{Numerical Experiments}\label{sec:experiment}

In this section, we study the RST-MRF priors from both a qualitative and quantitative perspective, visually and in terms of runtime. We consider three different standard inverse imaging problems of increasing ill-posedness, all on images of size 128-by-128. Specifically, we consider a denoising problem with a relatively large amount of noise, see Figure \ref{fig:Denoising-data}, two deblurring problems, see Figure \ref{fig:deblurring-data}, and an inpainting problem, see Figure \ref{fig:inpainting-data}.

\begin{figure}
    \centering
    \begin{subfigure}[b]{1.0\textwidth} 
        \centering
        \includegraphics[width=0.35\linewidth]{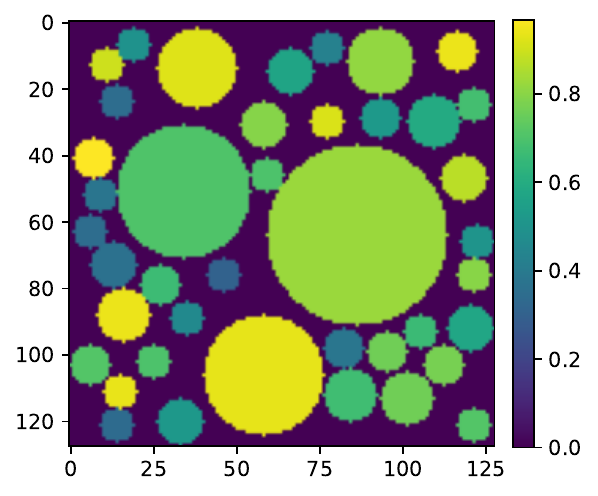}
        \includegraphics[width=0.35\linewidth]{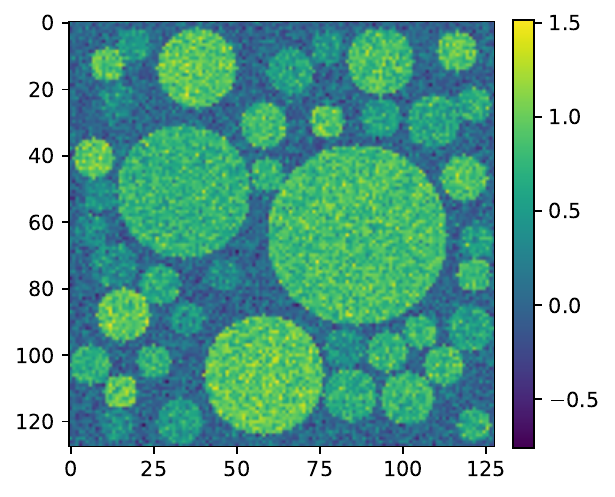}
        \caption{Ground truth and data of the denoising problem. Additive, independent Gaussian noise has been added with componentwise standard deviation $\sigma = 0.2$.}
        \label{fig:Denoising-data}
    \end{subfigure}
    \begin{subfigure}[b]{1.0\textwidth} 
        \centering
        \includegraphics[width=0.35\linewidth]{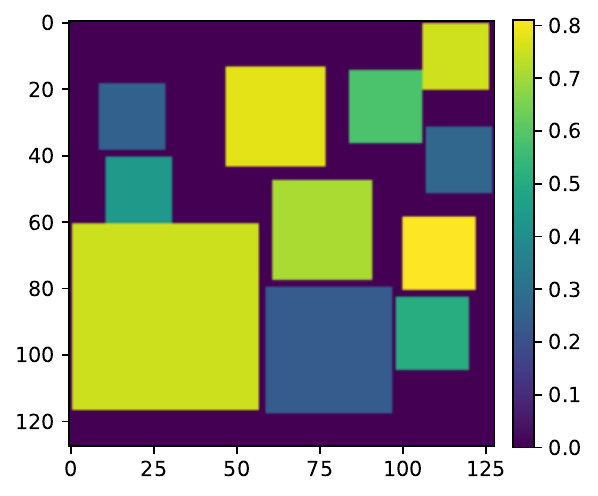}
        \includegraphics[width=0.35\linewidth]{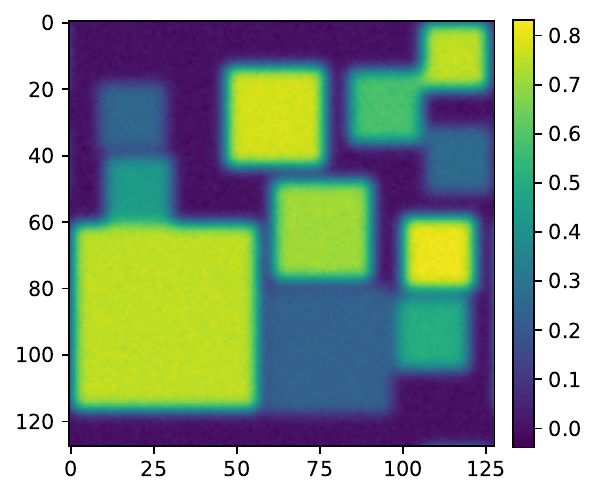}
        \includegraphics[width=0.35\linewidth]{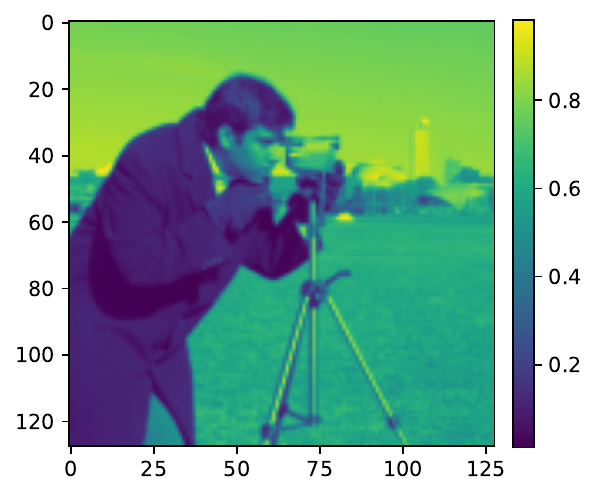}
        \includegraphics[width=0.35\linewidth]{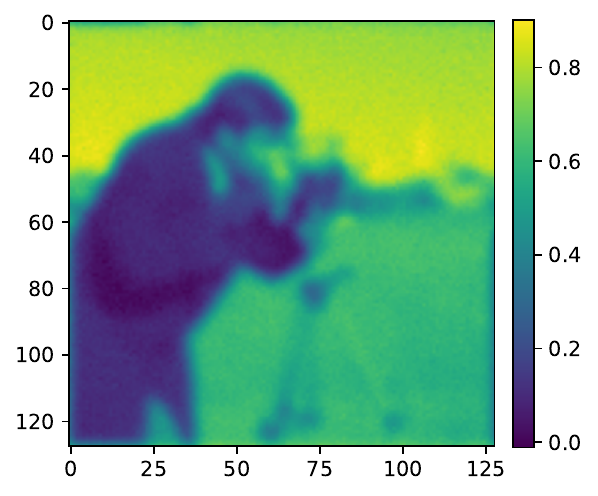}
        \caption{Ground truths and data of the deblurring problems. Additive, independent Gaussian noise has been added with componentwise standard deviation $\sigma = 10^{-2}$.}
        \label{fig:deblurring-data}
    \end{subfigure}
    \begin{subfigure}[b]{1.0\textwidth} 
        \centering
        \includegraphics[width=0.35\linewidth]{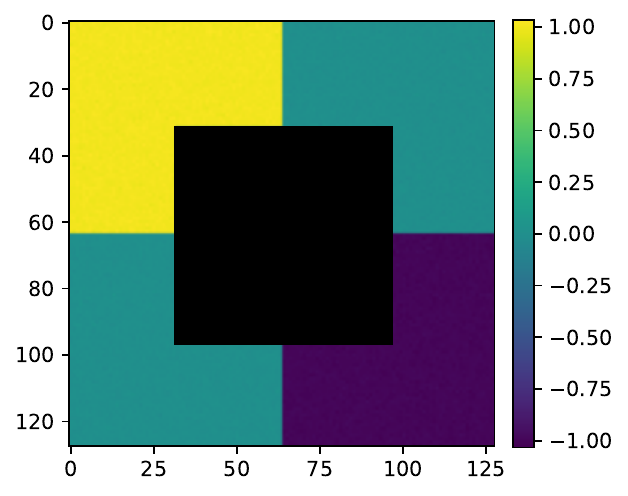}
        \caption{Data of the noisy inpainting problem. The whole image is 128-by-128, with the center 64-by-64 square removed. Additive, independent Gaussian noise has been added with componentwise standard deviation $\sigma = 10^{-2}$.}
        \label{fig:inpainting-data}
    \end{subfigure}
\end{figure}

\subsubsection*{Various considerations}
For both GMRF, LMRF and their RST variants, the prior strength $\lambda$ is used such that the density of the difference is proportional to $\phi\left(\lambda (x_{v_1} - x_{v_2})\right)$ such that larger $\lambda$ pushes the difference more towards zero. For the CMRF and its RST variants, the prior strength $\lambda$ is used such that the density of the difference is proportional to $\phi\left(\lambda^{-1} (x_{v_1} - x_{v_2})\right)$ such that larger $\lambda$ puts more weight on the heavy tail allowing for larger jumps. Thereby, for all priors, a larger $\lambda$ will have a stronger edge-preserving effect, as can be seen in the figures throughout this section.

To sample from the conditional posterior $\vec{x} \,|\, T, \vec{y}$, we use scale-mixture representations when applicable, and sample from the resulting Gaussian posterior by solving a randomized linear least-squares problem, see for example \cite{bardsley2012mcmc}. To solve this optimization problem, we use the Conjugate Gradient (CG) algorithm as implemented in \texttt{SciPy} \cite{2020SciPy-NMeth} with a relative tolerance of $10^{-6}$. The linear system can be severely ill-conditioned, especially for the RST-LMRF and RST-CMRF due to their heavier tails and/or sharp peaks around zero, hence some experiments make use of an approximate Jacobi preconditioner with the diagonal matrix estimated using a variant of Hutchinson's algorithm \cite{bekas2007estimator}. 

Any means or standard deviations for the non-log-concave priors CMRF, RST-GMRF, RST-LMRF and RST-CMRF are computed as an average over 100 independent runs. For convenience, we denote by $\rho$ the original termination weight divided by $\phi(0)$, such that the runtime of Wilson's algorithm is bounded by $\mathcal{O}\left(\rho^{-1}|V|\right)$. The full Python code of these experiments is available\footnote{The code can be found at: \url{https://github.com/jeverink/RST-MRF-paper}}.

\subsection{Denoising}

\begin{figure}[htbp]
\centering

% Row 1
\begin{subfigure}{0.48\textwidth}
    \centering
    \includegraphics[width=\linewidth]{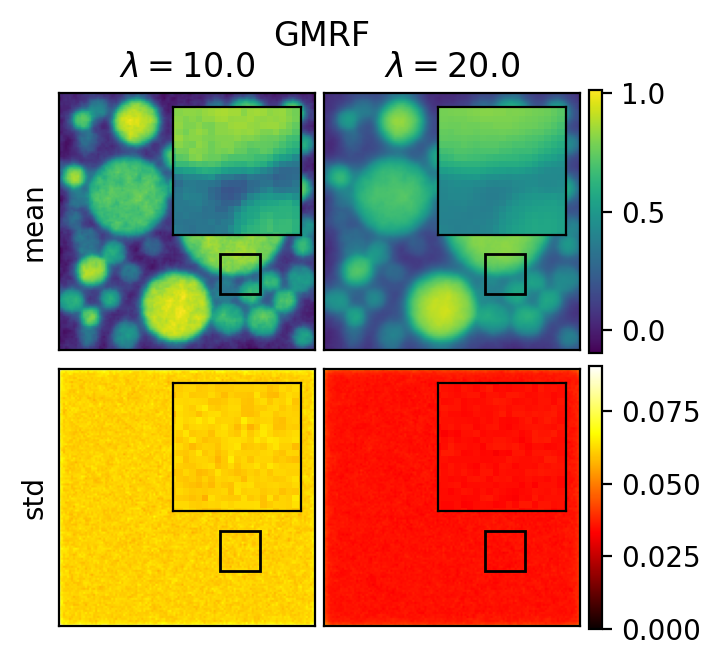}
\end{subfigure}
\hfill
\begin{subfigure}{0.48\textwidth}
    \centering
    \includegraphics[width=\linewidth]{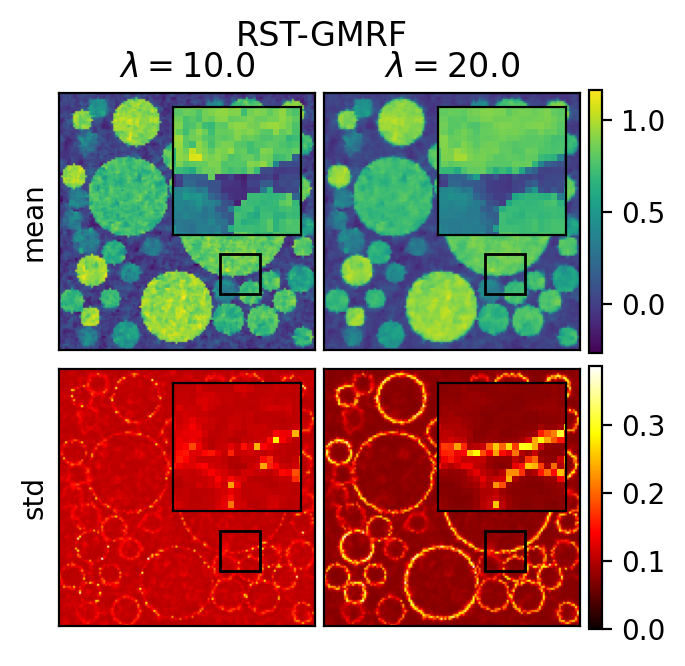}
\end{subfigure}

%\vspace{0.5em}

% Row 2
\begin{subfigure}{0.48\textwidth}
    \centering
    \includegraphics[width=\linewidth]{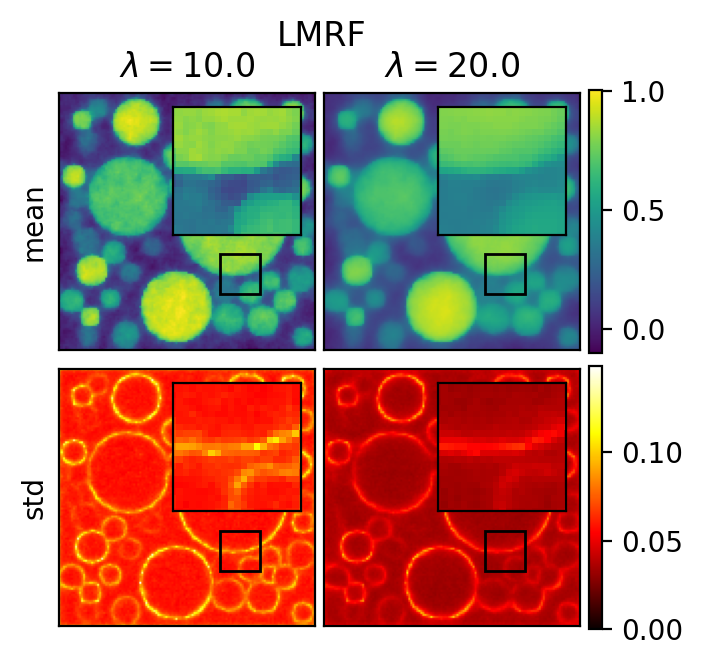}
\end{subfigure}
\hfill
\begin{subfigure}{0.48\textwidth}
    \centering
    \includegraphics[width=\linewidth]{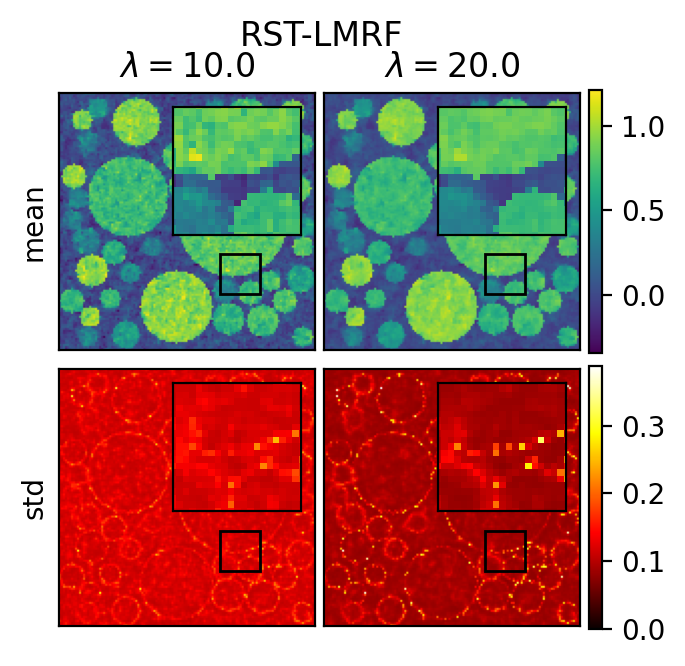}
\end{subfigure}

%\vspace{0.5em}

% Row 3
\begin{subfigure}{0.48\textwidth}
    \centering
    \includegraphics[width=\linewidth]{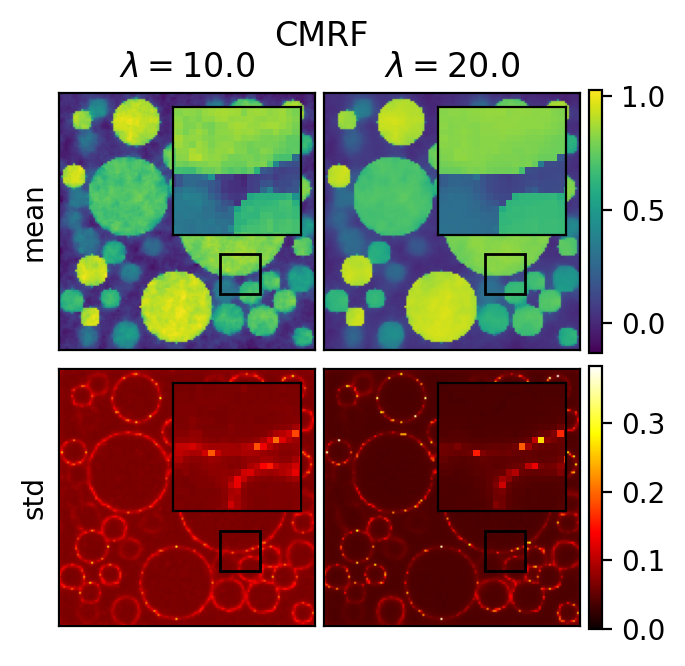}
\end{subfigure}
\hfill
\begin{subfigure}{0.48\textwidth}
    \centering
    \includegraphics[width=\linewidth]{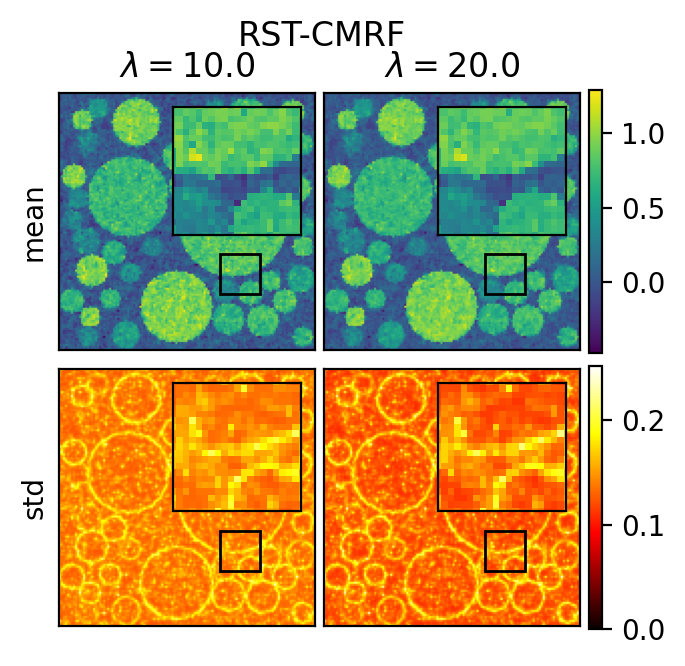}
\end{subfigure}

\caption{Componentwise means and standard deviations for the various MRFs and RST-MRFs.}
\label{fig:denoising_mean_std}
\end{figure}
Figure \ref{fig:denoising_mean_std} shows componentwise means and standard deviations of the denoising problem with GMRF, LMRF, CMRF, RST-GMRF, RST-LMRF and RST-CMRF priors for two prior strengths. All priors, except for the GMRF and LMRF which oversmooth, seem to preserve edges, signified by an increased standard deviation around the edges. Whilst the LMRF and CMRF preserve sufficiently large edges, it is noticeable that smaller edges get smoothed out similarly to the GMRF in the mean images. While all the RST hyperpriors preserve edges, it is clear that they do smooth the piecewise constant regions less, a result of the lack of connecting edges in these regions.

\begin{figure}
    \centering
    \includegraphics[width=1.0\linewidth]{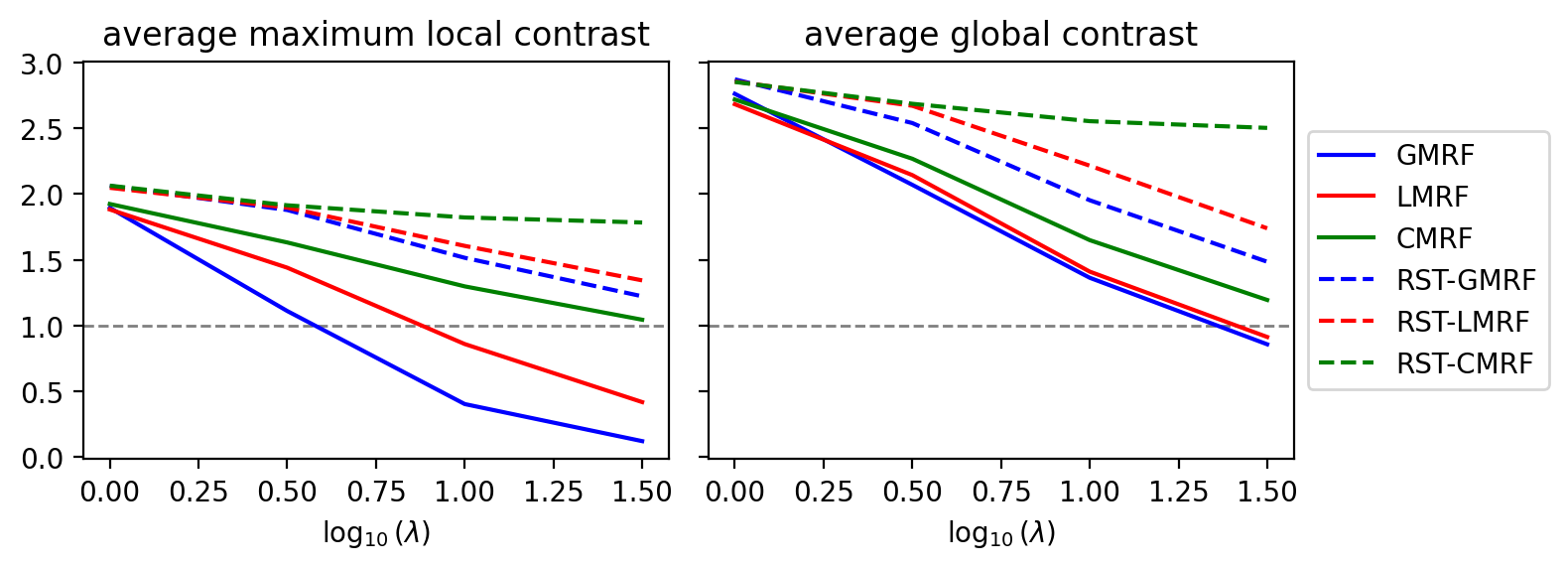}
    \caption{Average contrasts in the denoising problem for the various prior models at various strengths. Maximum local contrast refers to the largest absolute difference between neighbouring pixels. Global contrast refers to the difference between the largest and smallest value. The dashed gray line at $1.0$ corresponds to the ground truth.}
    \label{fig:denoising_contrast}
\end{figure}

To highlight the effect that the RST hyperprior has on the contrast in the images, Figure \ref{fig:denoising_contrast} shows two contrast measures as function of the prior strength $\lambda$ for all six considered priors. In all of the cases the standard MRF priors have lower contrast compared to their RST variants and heavier tailed distributions have higher contrast as expected.

\subsection{Deblurring}

\begin{figure}[htbp]
\centering

% Row 1
\begin{subfigure}{0.48\textwidth}
    \centering
    \includegraphics[width=\linewidth]{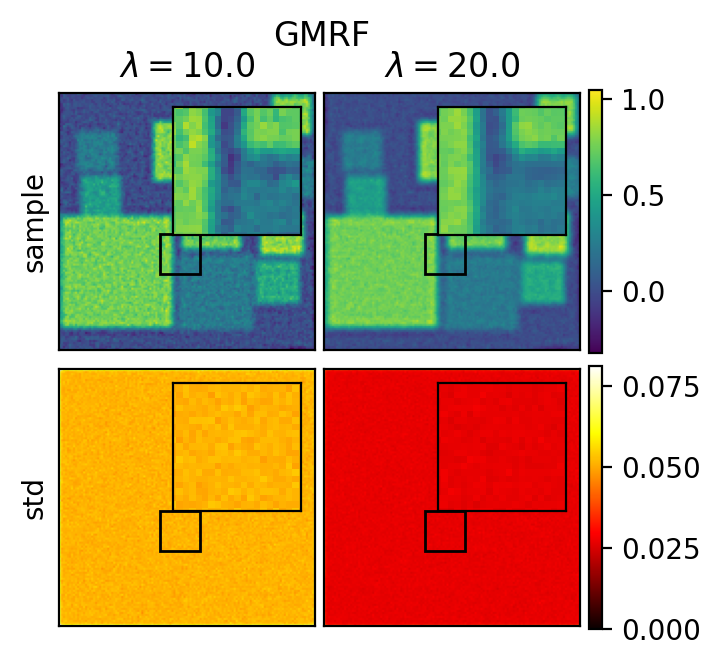}
\end{subfigure}
\hfill
\begin{subfigure}{0.48\textwidth}
    \centering
    \includegraphics[width=\linewidth]{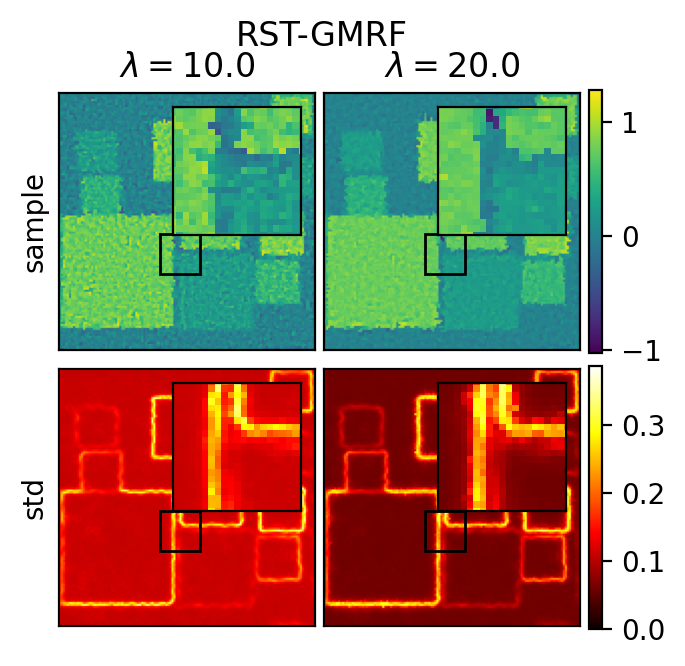}
\end{subfigure}

%\vspace{0.5em}

% Row 2
\begin{subfigure}{0.48\textwidth}
    \centering
    \includegraphics[width=\linewidth]{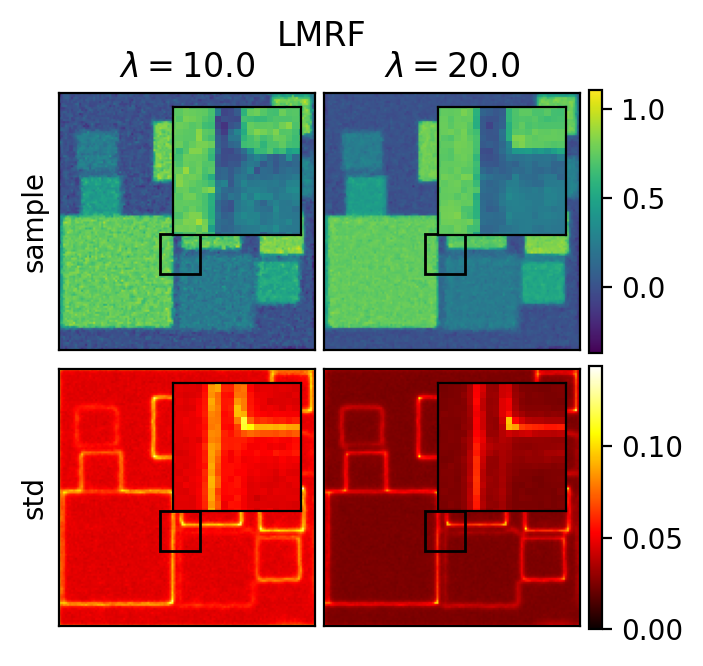}
\end{subfigure}
\hfill
\begin{subfigure}{0.48\textwidth}
    \centering
    \includegraphics[width=\linewidth]{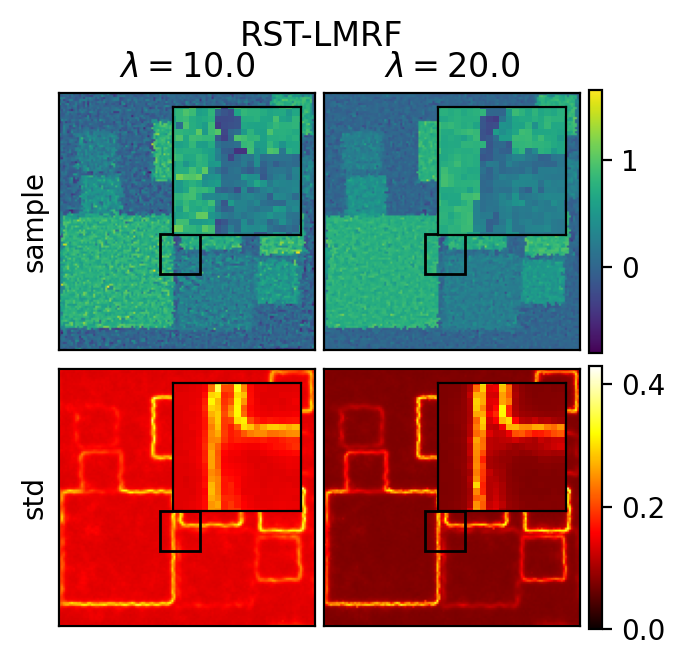}
\end{subfigure}

%\vspace{0.5em}

% Row 3
\begin{subfigure}{0.48\textwidth}
    \centering
    \includegraphics[width=\linewidth]{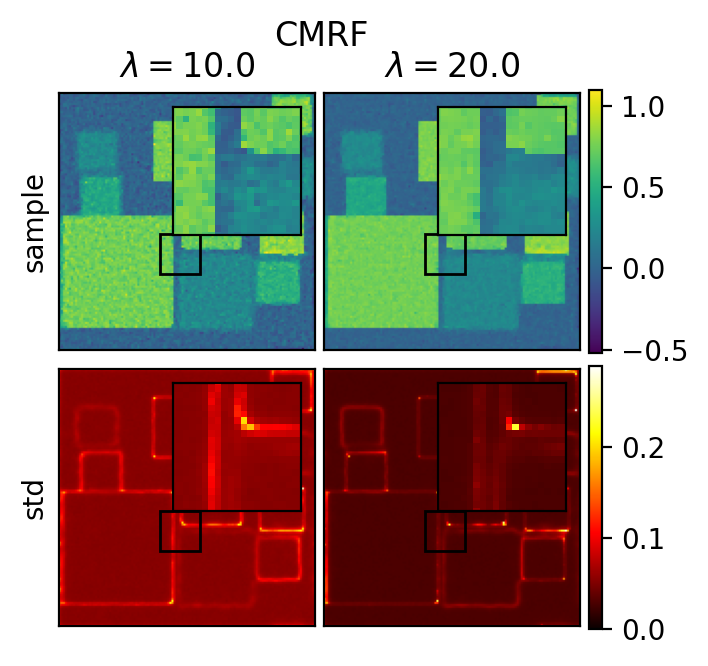}
\end{subfigure}
\hfill
\begin{subfigure}{0.48\textwidth}
    \centering
    \includegraphics[width=\linewidth]{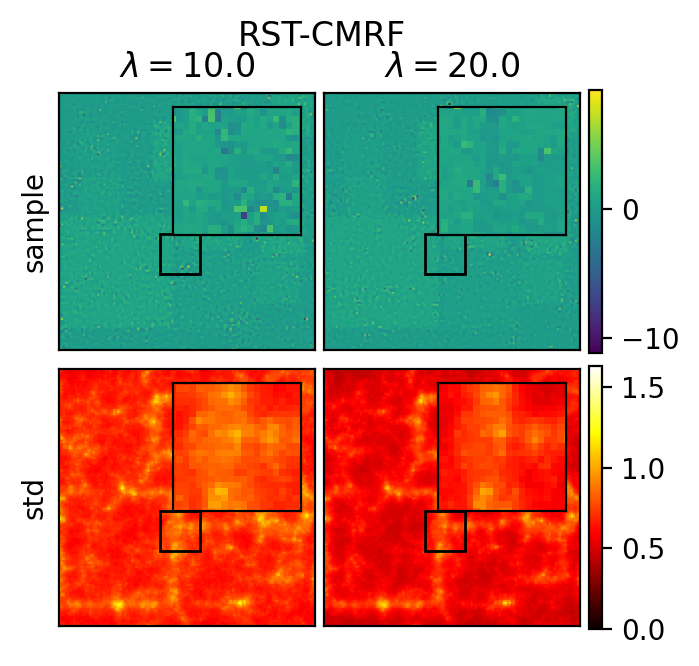}
\end{subfigure}

\caption{Samples and componentwise standard deviations for the various MRFs and RST-MRFs.}
\label{fig:deblurring_results}
\end{figure}

\begin{figure}[htbp]
\centering

% Row 1
\begin{subfigure}{0.48\textwidth}
    \centering
    \includegraphics[width=\linewidth]{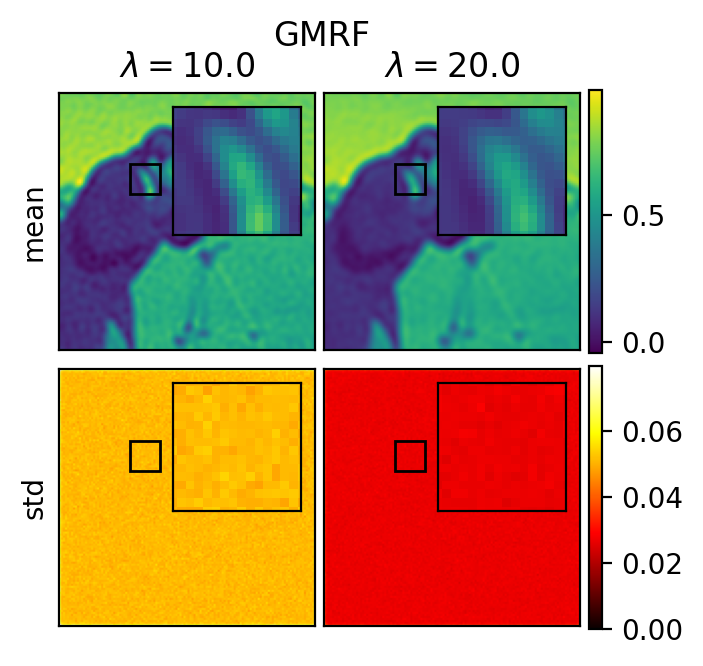}
\end{subfigure}
\hfill
\begin{subfigure}{0.48\textwidth}
    \centering
    \includegraphics[width=\linewidth]{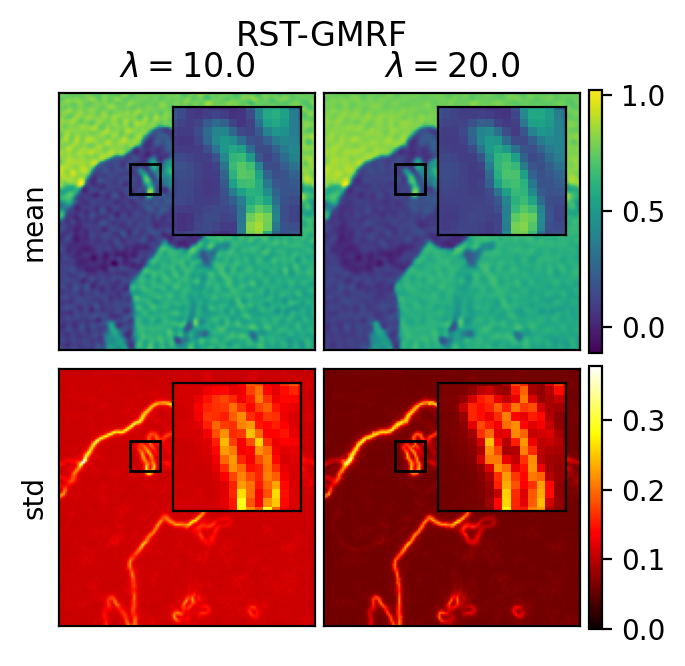}
\end{subfigure}

%\vspace{0.5em}

% Row 2
\begin{subfigure}{0.48\textwidth}
    \centering
    \includegraphics[width=\linewidth]{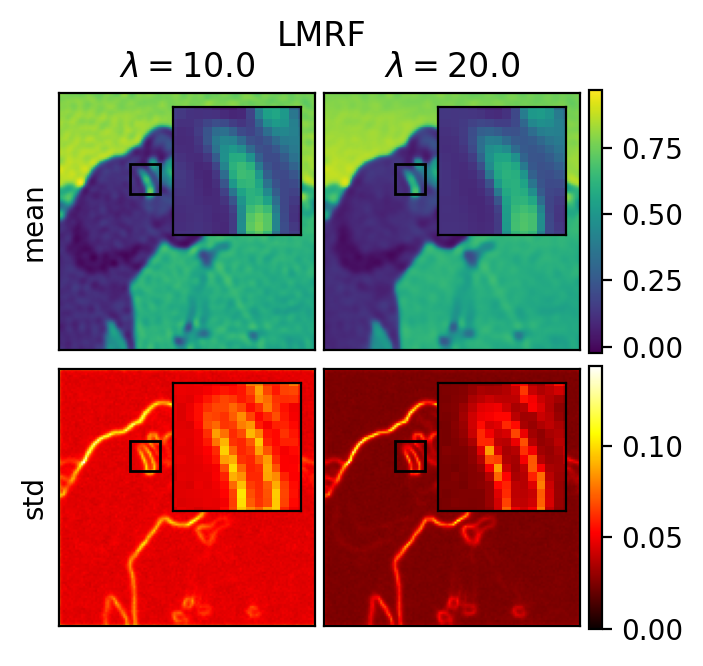}
\end{subfigure}
\hfill
\begin{subfigure}{0.48\textwidth}
    \centering
    \includegraphics[width=\linewidth]{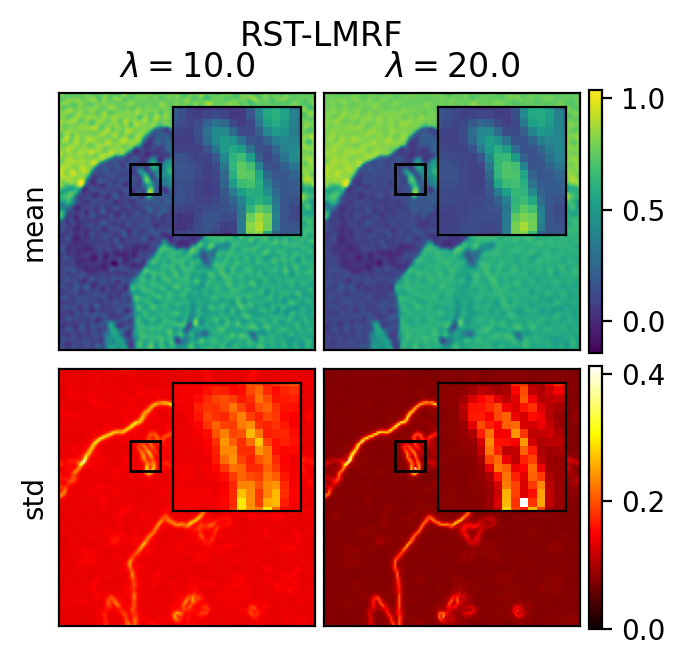}
\end{subfigure}

%\vspace{0.5em}

% Row 3
\begin{subfigure}{0.48\textwidth}
    \centering
    \includegraphics[width=\linewidth]{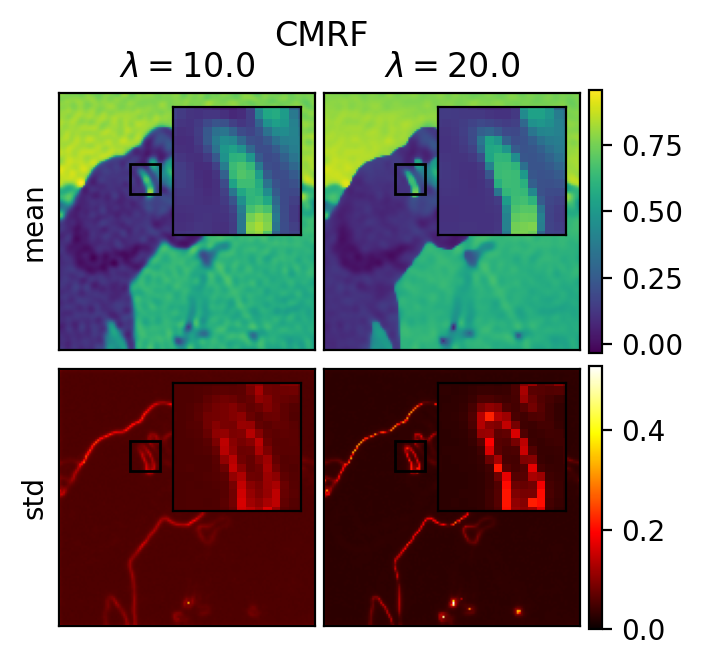}
\end{subfigure}
\hfill
\begin{subfigure}{0.48\textwidth}
    \centering
    \includegraphics[width=\linewidth]{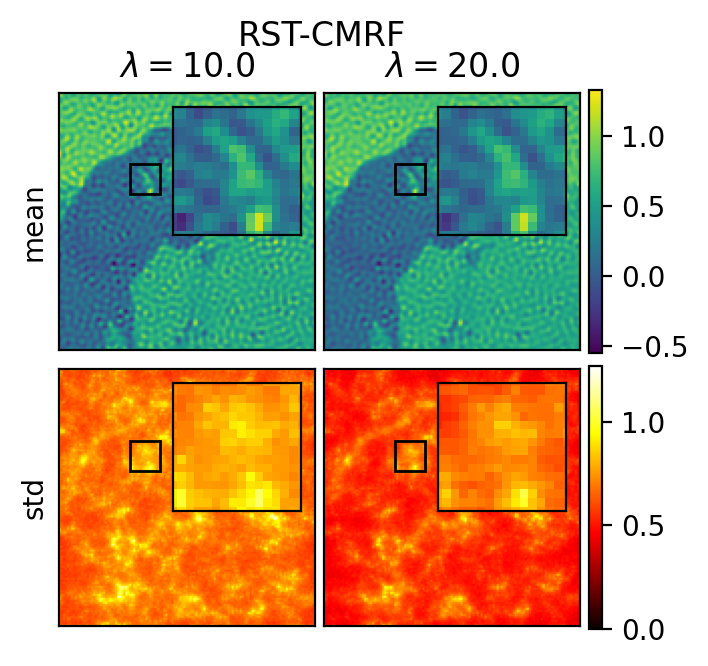}
\end{subfigure}

\caption{Componentwise means and standard deviations for the various MRFs and RST-MRFs.}
\label{fig:deblurring_results_2}
\end{figure}

Figures \ref{fig:deblurring_results} and \ref{fig:deblurring_results_2} show samples/mean and standard deviations of the deblurring problems with all six prior models at multiple prior strengths. The RST-CMRF prior did not result in any useful results at various prior strengths and algorithmic parameters. Whilst the ground truth in Figure \ref{fig:deblurring-data} clearly shows straight edges for the first problem, it is reasonable to expect some jagged edges when looking at the blurred edges. As seen in Figure \ref{fig:deblurring_results}, the results for the GMRF, LMRF and CMRF show a larger tendency towards straight edges, whilst the RST-GMRF and RST-LMRF priors naturally have some roughness in their samples. This roughness appears in the samples, but can average out in the mean and look smooth again, though with increased uncertainty in the standard deviation around the edges. This effect is shown in Figure \ref{fig:deblurring_results_2}.

\begin{figure}
    \centering
    \begin{subfigure}{0.32\textwidth}
        \centering
        \includegraphics[width=\linewidth]{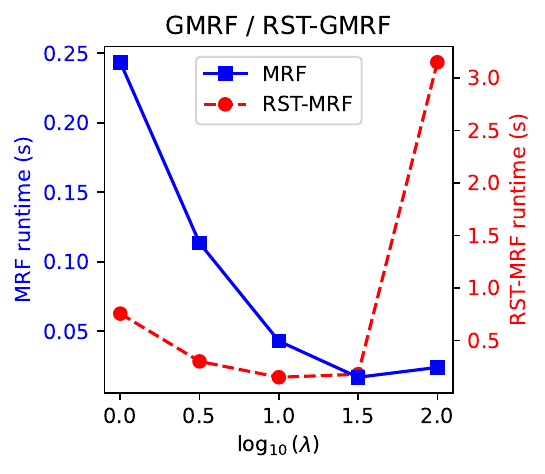}
    \end{subfigure}
    \hfill
    \begin{subfigure}{0.32\textwidth}
        \centering
        \includegraphics[width=\linewidth]{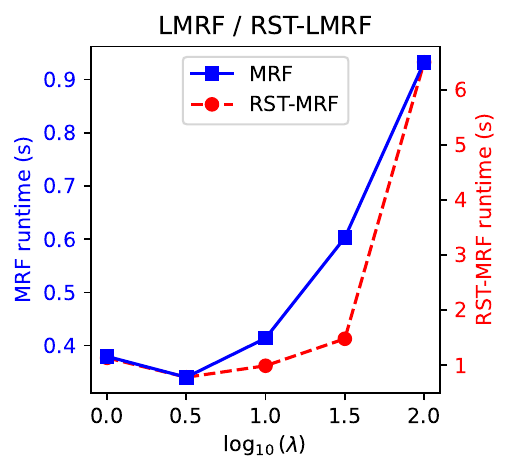}
    \end{subfigure}
    \hfill
    \begin{subfigure}{0.32\textwidth}
        \centering
        \includegraphics[width=\linewidth]{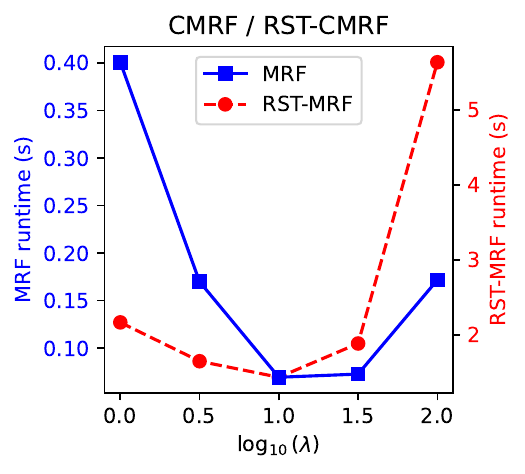}
    \end{subfigure}
    \caption{Average runtime per sample for the first deblurring experiment at different prior strengths $\lambda$. Note the different y-axis for the MRF and RST-MRF priors.}
    \label{fig:deblurring:runtime}
\end{figure}

The deblurring problem is inherently ill-conditioned. However, when using Linear-RTO for sampling from (conditional) Gaussian posterior distributions, the prior can greatly improve the conditioning of the system and improve the efficiency of the conjugate gradient solver. Figure \ref{fig:deblurring:runtime} shows the runtime for all the proposed priors at various prior strengths. We do not consider preconditioning and the cost of sampling random weighted spanning trees is only a small part of the computation time, hence the runtime serves as a proxy for the conditioning of the linear system solved by CG. The MRF priors are frequently faster than the RST-MRF priors, mainly due to the additional edges in the MRF priors improving the conditioning of the linear system and the additional overhead of running Gibbs samplers. Laplace-based priors are generally slower than the others, both the LMRF and RST-LMRF. This is mainly caused by the scale-mixture representation causing a mixture of very small and large values, resulting in a conditional Gaussian that itself can again be quite ill-conditioned, unlike the CMRF and RST-CMRF priors, where the values have a large spread, but are not clustered at very low values. Comparing the visual quality to the runtime of the RST-MRF priors suggests that the RST-GMRF offers the best runtime among the RST-MRF priors while retaining the random spanning tree behaviour.

\subsection{Inpainting}

\begin{figure}[htbp]
\centering

% Row 1
\begin{subfigure}{0.48\textwidth}
    \centering
    \includegraphics[width=\linewidth]{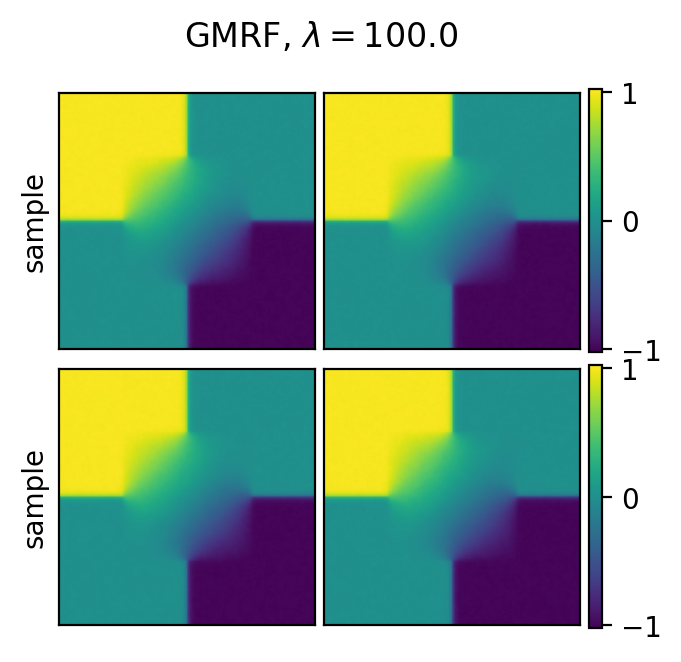}
\end{subfigure}
\hfill
\begin{subfigure}{0.48\textwidth}
    \centering
    \includegraphics[width=\linewidth]{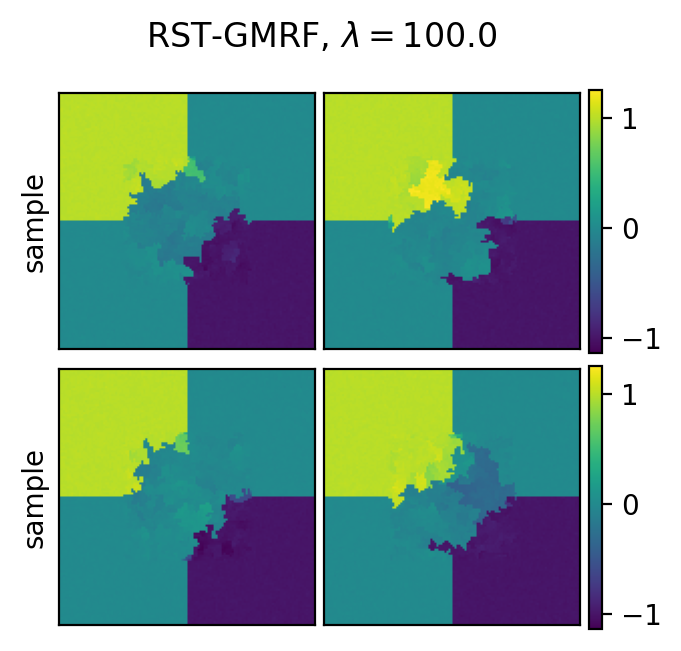}
\end{subfigure}

%\vspace{0.5em}

% Row 2
\begin{subfigure}{0.48\textwidth}
    \centering
    \includegraphics[width=\linewidth]{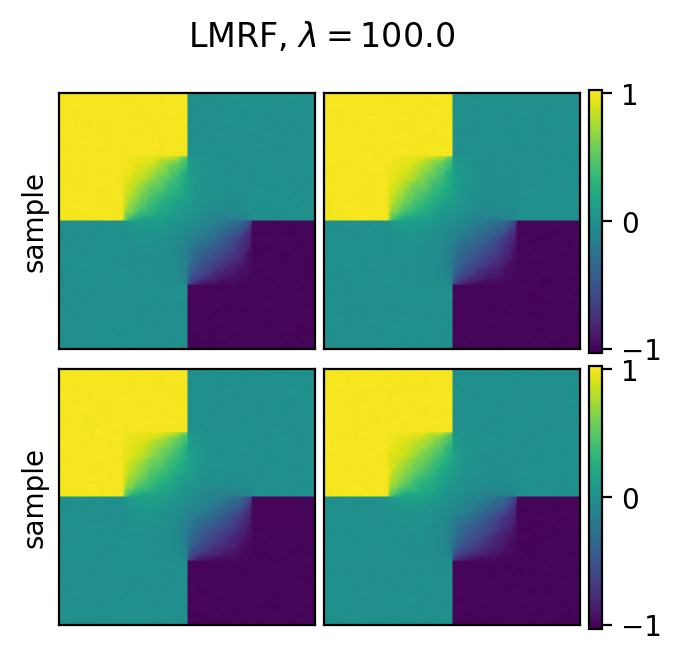}
\end{subfigure}
\hfill
\begin{subfigure}{0.48\textwidth}
    \centering
    \includegraphics[width=\linewidth]{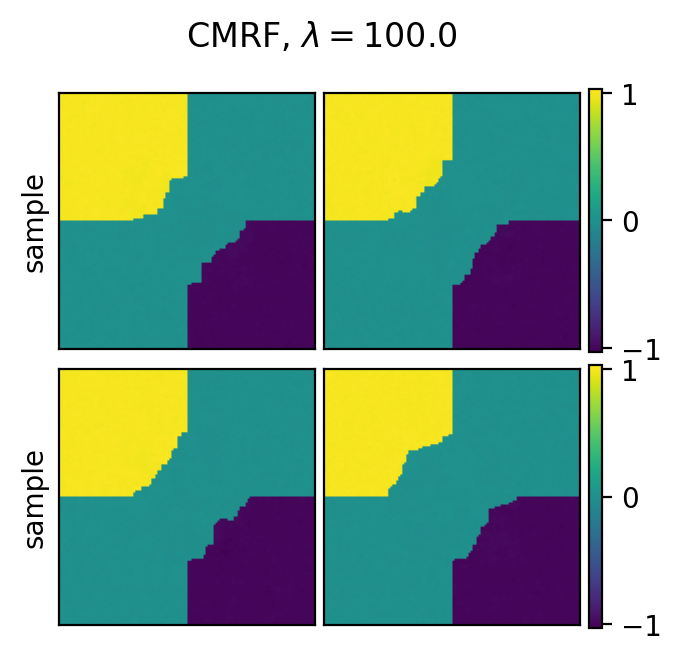}
\end{subfigure}

\caption{Results of various MRFs and RST-MRFs for the inpainting problem.}
\label{fig:inpainting_results}
\end{figure}

Figure \ref{fig:inpainting_results} shows samples of the inpainting problem with GMRF, LMRF, CMRF and RST-GMRF. The RST-LMRF and RST-CMRF samplers are not considered because their numerical instability. The edge-preserving CMRF, RST-GMRF both try to find some curve between the disconnected boundaries of the low and high intensity regions. For the CMRF, at sufficiently large prior strength, the most likely curves are monotonic from one corner to the other, such that the jumps happen over the least number of edges. Whilst in the case of total variation regularization, this would be a short line, here any monotonic curve would work, with the piecewise constant regions resulting in almost insignificant weight compared to the edges. For the RST-GMRF, at most one edge along the connecting curve is generally penalized, often none due to the positive termination weight. Hence any connecting curve will be equally likely, with fractal-like interfaces due to the nature of loop-erased random walks as discussed around Theorem \ref{thm:fractal_interface}.

%\input{tex_figures/inpainting_traces}
%Although the GMRF and LMRF result in log-concave posterior distributions under standard Gaussian likelihood assumptions, the CMRF, RST-GMRF and RST-LMRF are generally not log-concave and the samplers will easily get stuck in the mode corresponding to a particular boundary interface. The Gibbs sampler is theoretically convergent, but it is unlikely to escape these modes and it is not expected to be able to explore the full posterior in any reasonable amount of computation time. Moreover, the sampling algorithm might prefer to keep certain nearly piecewise constant regions connected, increasing the termination weight for the RST priors reduces the connectivity within these regions and makes the sampler discover different boundaries at the cost of reducing some of the stability of the algorithm. Especially for the inpainting problem, if a connected component in the random spanning forest is completely contained within the removed region, that component can drift to almost any value as it is no longer connected to any known value on the boundary. This becomes increasingly more likely as the termination weight increases, increasing the chance of instability. Figure \ref{fig:inpainting_traces} shows the trace of the sampler for the RST-GMRF and RST-LMRF priors at five different locations for different termination weights $\rho$. It clearly shows that increasing the termination weight $\rho$ helps exploration at the cost of stability.

\section{Conclusions}\label{sec:conclusions}

In this work, we presented random spanning trees as a hyperprior on the graph structure of multiple Markov random field priors on two-dimensional images. This combination gives rise to prior distributions with fractal-like interfaces within the image samples, whilst not over-smoothing other parts of the image. We propose an efficient Gibbs sampler for exploring the posterior distributions using this prior, although the multi-modal nature of the prior does give the sampler the tendency to get stuck in modes. Whilst a spanning tree provides the lowest number of edges required for making the graph connected, this reduction impacts the stability of sampling and the noise-reducing effect of Markov random fields. This computational issue is especially the case for the random spanning tree hyperprior combined with the Laplace and Cauchy Markov random fields, where the optimization problem we use for sampling from the unknown image is severely ill-conditioned.

The focus of this work has been on sampling and uncertainty quantification. Whilst mathematical convergence of the chain generated by the Gibbs sampler is guaranteed, the multi-modal nature of the posterior distributions makes full exploration of the posterior still problematic. Similar issues would arise in estimating the maximum a posterior (MAP) estimate using block coordinate descent algorithms, which does not have the convergence guarantee that the Gibbs sampler has. Alternative methods based on joint optimization of the continuous and discrete variables, such as through mixed-integer quadratic programming \cite{lazimy1982mixed}, could be studied to improve efficiency. Finding approaches for improving the stability of sampling could also be beneficial for scaling the method to higher resolutions, where the current bottleneck lies in sampling from the continuous random image.

The random spanning tree Markov random field priors have the benefit of efficient sampling made possible in the finite-dimensional setting. The fractal-like interfaces appearing in the discretization limit make the limiting distribution an interesting direction for future research. For example, it is known that the LMRF is not edge-preserving when the discretization is sufficiently fine and an infinite dimensional total variation prior is not known to exist  \cite{lassas2004can}. The limiting distribution of random spanning trees has been extensively studied in the literature, see e.g., \cite{schramm2009contour, lawler2011conformal}, but the limiting behaviour of the prior and posterior when incorporated into random spanning tree Markov random field priors is not yet known.

\section*{Acknowledgments}This work was supported by the Research Council of Finland Flagship of Advanced Mathematics for Sensing, Imaging and Modelling (Grant No. 358944). Initial work started as part of the CUQI project which is supported by The Villum Foundation (Grant No. 25893).

\bibliographystyle{abbrv}
\bibliography{references}

@article{lassas2004can,
  title={Can one use total variation prior for edge-preserving {Bayesian} inversion?},
  author={Lassas, Matti and Siltanen, Samuli},
  journal={Inverse problems},
  volume={20},
  number={5},
  pages={1537},
  year={2004},
  publisher={IOP Publishing}
}

@article{kekkonen2023random,
  title={Random tree {Besov} priors--Towards fractal imaging},
  author={Kekkonen, Hanne and Lassas, Matti and Saksman, Eero and Siltanen, Samuli},
  journal={Inverse Problems and Imaging},
  volume={17},
  number={2},
  pages={507--531},
  year={2023},
  publisher={Inverse Problems and Imaging}
}

@inproceedings{wilson1996generating,
  title={Generating random spanning trees more quickly than the cover time},
  author={Wilson, David Bruce},
  booktitle={Proceedings of the twenty-eighth annual ACM symposium on Theory of computing},
  pages={296--303},
  year={1996}
}

@article{propp1998get,
  title={How to get a perfectly random sample from a generic {Markov} chain and generate a random spanning tree of a directed graph},
  author={Propp, James Gary and Wilson, David Bruce},
  journal={Journal of Algorithms},
  volume={27},
  number={2},
  pages={170--217},
  year={1998},
  publisher={Elsevier}
}

@book{lyons2017probability,
  title={Probability on trees and networks},
  author={Lyons, Russell and Peres, Yuval},
  volume={42},
  year={2017},
  publisher={Cambridge University Press}
}

@article{prim1957shortest,
  title={Shortest connection networks and some generalizations},
  author={Prim, Robert Clay},
  journal={The Bell System Technical Journal},
  volume={36},
  number={6},
  pages={1389--1401},
  year={1957},
  publisher={Nokia Bell Labs}
}

@article{kruskal1956shortest,
  title={On the shortest spanning subtree of a graph and the traveling salesman problem},
  author={Kruskal, Joseph B},
  journal={Proceedings of the American Mathematical society},
  volume={7},
  number={1},
  pages={48--50},
  year={1956},
  publisher={JSTOR}
}

@article{uribe2023horseshoe,
  title={Horseshoe priors for edge-preserving linear {Bayesian} inversion},
  author={Uribe, Felipe and Dong, Yiqiu and Hansen, Per Christian},
  journal={SIAM Journal on Scientific Computing},
  volume={45},
  number={3},
  pages={B337--B365},
  year={2023},
  publisher={SIAM}
}

@incollection{lawler2011conformal,
  title={Conformal invariance of planar loop-erased random walks and uniform spanning trees},
  author={Lawler, Gregory F and Schramm, Oded and Werner, Wendelin},
  booktitle={Selected Works of Oded Schramm},
  pages={931--987},
  year={2011},
  publisher={Springer}
}

@article{suuronen2022cauchy,
  title={{Cauchy} {Markov} random field priors for {Bayesian} inversion},
  author={Suuronen, Jarkko and Chada, Neil K and Roininen, Lassi},
  journal={Statistics and computing},
  volume={32},
  number={2},
  pages={33},
  year={2022},
  publisher={Springer}
}

@article{kolmogorov2016total,
  title={Total variation on a tree},
  author={Kolmogorov, Vladimir and Pock, Thomas and Rolinek, Michal},
  journal={SIAM Journal on Imaging Sciences},
  volume={9},
  number={2},
  pages={605--636},
  year={2016},
  publisher={SIAM}
}

@article{siltanen2003statistical,
  title={Statistical inversion for medical x-ray tomography with few radiographs: I. {General} theory},
  author={Siltanen, Samuli and Kolehmainen, Ville and J{\"a}rvenp{\"a}{\"a}, Seppo and Kaipio, Jari P and Koistinen, Petri and Lassas, Matti and Pirttil{\"a}, Juha and Somersalo, Erkki},
  journal={Physics in Medicine \& Biology},
  volume={48},
  number={10},
  pages={1437},
  year={2003},
  publisher={IOP Publishing}
}

@article{markkanen2019cauchy,
  title={{Cauchy} difference priors for edge-preserving {Bayesian} inversion},
  author={Markkanen, Markku and Roininen, Lassi and Huttunen, Janne MJ and Lasanen, Sari},
  journal={Journal of Inverse and Ill-posed Problems},
  volume={27},
  number={2},
  pages={225--240},
  year={2019}
}

@article{lazimy1982mixed,
  title={Mixed-integer quadratic programming},
  author={Lazimy, Rafael},
  journal={Mathematical Programming},
  volume={22},
  number={1},
  pages={332--349},
  year={1982},
  publisher={Springer}
}

@book{rue2005gaussian,
  title={{Gaussian} {Markov} random fields: theory and applications},
  author={Rue, Havard and Held, Leonhard},
  year={2005},
  publisher={Chapman and Hall/CRC}
}

@article{duan2023bayesian,
  title={{Bayesian} spanning tree: estimating the backbone of the dependence graph},
  author={Duan, Leo L and Dunson, David B},
  journal={Journal of Machine Learning Research},
  volume={24},
  number={397},
  pages={1--44},
  year={2023}
}

@article{schreck2015shrinkage,
  title={A shrinkage-thresholding {Metropolis} adjusted {Langevin} algorithm for {Bayesian} variable selection},
  author={Schreck, Amandine and Fort, Gersende and Le Corff, Sylvain and Moulines, Eric},
  journal={IEEE Journal of Selected Topics in Signal Processing},
  volume={10},
  number={2},
  pages={366--375},
  year={2015},
  publisher={IEEE}
}

@article{hans2007shotgun,
  title={Shotgun stochastic search for “large p” regression},
  author={Hans, Chris and Dobra, Adrian and West, Mike},
  journal={Journal of the American Statistical Association},
  volume={102},
  number={478},
  pages={507--516},
  year={2007},
  publisher={Taylor \& Francis}
}

@book{kaipio2005statistical,
  title={Statistical and computational inverse problems},
  author={Kaipio, Jari P and Somersalo, Erkki},
  year={2005},
  publisher={Springer}
}

@article{kekkonen2026random,
  title={Random tree {Besov} priors: {Data-driven} regularisation parameter selection},
  author={Kekkonen, Hanne and Tataris, Andreas},
  journal={arXiv preprint arXiv:2601.12957},
  year={2026}
}

@article{horst2025uncertainty,
  title={Uncertainty Quantification for Linear Inverse Problems with {Besov} Prior: A Randomize-Then-Optimize Method},
  author={Horst, Andreas and Maboudi Afkham, Babak and Dong, Yiqiu and Lemvig, Jakob},
  journal={Statistics and Computing},
  volume={35},
  number={4},
  pages={101},
  year={2025},
  publisher={Springer}
}

@article{brown2021sampling,
  title={Sampling strategies for fast updating of {Gaussian} {Markov} random fields},
  author={Brown, D Andrew and McMahan, Christopher S and Watson Self, Stella},
  journal={The American Statistician},
  volume={75},
  number={1},
  pages={52--65},
  year={2021},
  publisher={Taylor \& Francis}
}

@article{bardsley2012laplace,
  title={{Laplace-distributed} increments, the {Laplace} prior, and edge-preserving regularization.},
  author={Bardsley, Johnathan M},
  journal={Journal of Inverse \& Ill-Posed Problems},
  volume={20},
  number={3},
  year={2012}
}

@article{bardsley2012mcmc,
  title={MCMC-based image reconstruction with uncertainty quantification},
  author={Bardsley, Johnathan M},
  journal={SIAM Journal on Scientific Computing},
  volume={34},
  number={3},
  pages={A1316--A1332},
  year={2012},
  publisher={SIAM}
}

@article{sheffield2007gaussian,
  title={{Gaussian} free fields for mathematicians},
  author={Sheffield, Scott},
  journal={Probability theory and related fields},
  volume={139},
  number={3},
  pages={521--541},
  year={2007},
  publisher={Springer}
}

@article{schramm2009contour,
  title={Contour lines of the two-dimensional discrete {Gaussian} free field},
  author={Schramm, Oded and Sheffield, Scott},
  journal={Acta Math},
  volume={202},
  pages={21--137},
  year={2009}
}

@article{senchukova2024bayesian,
  title={{Bayesian} inversion with {Students’ t} priors based on {Gaussian} scale mixtures},
  author={Senchukova, Angelina and Uribe, Felipe and Roininen, Lassi},
  journal={Inverse Problems},
  volume={40},
  number={10},
  pages={105013},
  year={2024},
  publisher={IOP Publishing}
}

@article{tam2025exact,
  title={Exact sampling of spanning trees via fast-forwarded random walks},
  author={Tam, Edric and Dunson, David B and Duan, Leo L},
  journal={Biometrika},
  volume={112},
  number={2},
  pages={asaf031},
  year={2025},
  publisher={Oxford University Press}
}

@inproceedings{aleliunas1979random,
  title={Random walks, universal traversal sequences, and the complexity of maze problems},
  author={Aleliunas, Romas and Karp, Richard M and Lipton, Richard J and Lov{\'a}sz, L{\'a}szl{\'o} and Rackoff, Charles},
  booktitle={Proceedings of the 20th Annual Symposium on Foundations of Computer Science},
  pages={218--223},
  year={1979}
}

@article{aldous1989lower,
  title={Lower bounds for covering times for reversible {Markov} chains and random walks on graphs},
  author={Aldous, David J},
  journal={Journal of Theoretical Probability},
  volume={2},
  number={1},
  pages={91--100},
  year={1989},
  publisher={Springer}
}

@misc{aldous-fill-2014,
    AUTHOR = {Aldous, David and Fill, James Allen},
     TITLE = {Reversible Markov Chains and Random Walks on Graphs},
      YEAR = {2002},
      NOTE = {Unfinished monograph, recompiled 2014, available 
      at \url{http://www.stat.berkeley.edu/$\sim$aldous/RWG/book.html}}
      }

@article{shante1971introduction,
  title={An introduction to percolation theory},
  author={Shante, Vinod KS and Kirkpatrick, Scott},
  journal={Advances in Physics},
  volume={20},
  number={85},
  pages={325--357},
  year={1971},
  publisher={Taylor \& Francis}
}

@article{vono2022high,
  title={High-dimensional {Gaussian} sampling: A review and a unifying approach based on a stochastic proximal point algorithm},
  author={Vono, Maxime and Dobigeon, Nicolas and Chainais, Pierre},
  journal={SIAM Review},
  volume={64},
  number={1},
  pages={3--56},
  year={2022},
  publisher={SIAM}
}

@article{gelman2006prior,
  title={Prior distributions for variance parameters in hierarchical models (comment on article by {Browne} and {Draper})},
  author={Gelman, Andrew},
  journal={Bayesian Analysis},
  volume={1},
  pages={515--533},
  year={2006}
}

@article{flock2025continuous,
  title={Continuous {Gaussian} mixture solution for linear {Bayesian} inversion with application to {Laplace} priors},
  author={Flock, Rafael and Dong, Yiqiu and Uribe, Felipe and Zahm, Olivier},
  journal={Inverse Problems},
  volume={41},
  number={6},
  pages={065012},
  year={2025},
  publisher={IOP Publishing}
}

@ARTICLE{2020SciPy-NMeth,
  author  = {Virtanen, Pauli and Gommers, Ralf and Oliphant, Travis E. and
            Haberland, Matt and Reddy, Tyler and Cournapeau, David and
            Burovski, Evgeni and Peterson, Pearu and Weckesser, Warren and
            Bright, Jonathan and {van der Walt}, St{\'e}fan J. and
            Brett, Matthew and Wilson, Joshua and Millman, K. Jarrod and
            Mayorov, Nikolay and Nelson, Andrew R. J. and Jones, Eric and
            Kern, Robert and Larson, Eric and Carey, C J and
            Polat, {\.I}lhan and Feng, Yu and Moore, Eric W. and
            {VanderPlas}, Jake and Laxalde, Denis and Perktold, Josef and
            Cimrman, Robert and Henriksen, Ian and Quintero, E. A. and
            Harris, Charles R. and Archibald, Anne M. and
            Ribeiro, Ant{\^o}nio H. and Pedregosa, Fabian and
            {van Mulbregt}, Paul and {SciPy 1.0 Contributors}},
  title   = {{{SciPy} 1.0: Fundamental Algorithms for Scientific
            Computing in {Python}}},
  journal = {Nature Methods},
  year    = {2020},
  volume  = {17},
  pages   = {261--272},
  adsurl  = {https://rdcu.be/b08Wh},
  doi     = {10.1038/s41592-019-0686-2},
}

@article{park2008bayesian,
  title={The {Bayesian} lasso},
  author={Park, Trevor and Casella, George},
  journal={Journal of the american statistical association},
  volume={103},
  number={482},
  pages={681--686},
  year={2008},
  publisher={Taylor \& Francis}
}

@article{bekas2007estimator,
  title={An estimator for the diagonal of a matrix},
  author={Bekas, Costas and Kokiopoulou, Effrosyni and Saad, Yousef},
  journal={Applied numerical mathematics},
  volume={57},
  number={11-12},
  pages={1214--1229},
  year={2007},
  publisher={Elsevier}
}

@article{aldous1990random,
  title={The random walk construction of uniform spanning trees and uniform labelled trees},
  author={Aldous, David J},
  journal={SIAM Journal on Discrete Mathematics},
  volume={3},
  number={4},
  pages={450--465},
  year={1990},
  publisher={SIAM}
}

@inproceedings{broder1989generating,
  title={Generating random spanning trees},
  author={Broder, Andrei Z},
  booktitle={FOCS},
  volume={89},
  pages={442--447},
  year={1989}
}

@article{beffara2008dimension,
  title={The Dimension of the {SLE} Curves},
  author={Beffara, Vincent},
  journal={The Annals of Probability},
  pages={1421--1452},
  year={2008},
  publisher={JSTOR}
}

\end{document}